\documentclass[journal,twoside,web]{ieeecolor}
\usepackage{generic}
\usepackage{cite}
\usepackage{algorithm,algorithmic}
\usepackage{textcomp}

\usepackage{amsmath,amssymb,amsfonts}

\newtheorem{lemma}{Lemma}

\usepackage{graphicx}
\graphicspath{{Figures/}}
\usepackage{booktabs}
\usepackage{multirow}
\usepackage{tabularx}
\usepackage{makecell}
\usepackage{array} 
\newcolumntype{Y}{>{\raggedright\arraybackslash}X}

\usepackage{bbm}
\usepackage{url}
\usepackage[hidelinks]{hyperref} 

\usepackage{makecell}

\def\BibTeX{{\rm B\kern-.05em{\sc i\kern-.025em b}\kern-.08em
    T\kern-.1667em\lower.7ex\hbox{E}\kern-.125emX}}
\begin{document}
\title{Continuous Energy Landscape Model for Analyzing Brain State Transitions}
\author{Triet~M.~Tran,
        Seyed~Majid~Razavi,
        Dee~H.~Wu,
        and~Sina~Khanmohammadi*
\thanks{Triet~M.~Tran is with the Data Science and Analytics Institute, University of Oklahoma, Norman, OK 73109, USA, and also with the Pat Summitt Clinic, University of Tennessee Medical Center, Brain and Spine Institute, Knoxville, TN 37920, USA. }
\thanks{Seyed~Majid~Razavi is with the Data Science and Analytics Institute, University of Oklahoma, Norman, OK 73109, USA.}
\thanks{Dee~H.~Wu is with the Department of Radiological Sciences, University of Oklahoma Health Sciences Center, Oklahoma City, OK 73104, USA,
and also with the School of Computer Science, University of Oklahoma, Norman, OK 73109, USA.}
\thanks{Sina~Khanmohammadi is with the Data Science and Analytics Institute, University of Oklahoma, Norman, OK 73109, USA, and also with the School of Computer Science, University of Oklahoma, Norman, OK 73109, USA.\\
*Corresponding author: Sina~Khanmohammadi (e-mail: sinakhan@ou.edu).}}

\maketitle

\begin{abstract}
Energy landscape models characterize neural dynamics by assigning energy values to each brain state that reflect their stability or probability of occurrence. The conventional energy landscape models rely on binary brain state representation, where each region is considered either active or inactive based on some signal threshold. However, this binarization leads to significant information loss and an exponential increase in the number of possible brain states, making the calculation of energy values infeasible for large numbers of brain regions. To overcome these limitations, we propose a novel continuous energy landscape framework that employs Graph Neural Networks (GNNs) to learn a continuous precision matrix directly from functional MRI (fMRI) signals, preserving the full range of signal values during energy landscape computation. We validated our approach using both synthetic data and real-world fMRI datasets from brain tumor patients. Our results on synthetic data generated from a switching linear dynamical system (SLDS) and a Kuramoto model show that the continuous energy model achieved higher likelihood and more accurate recovery of basin geometry, state occupancy, and transition dynamics than conventional binary energy landscape models. In addition, results from the fMRI dataset indicate a $0.27$ increase in AUC for predicting working memory and executive function, along with a $0.35$ improvement in explained variance ($R^2$) for predicting reaction time. These findings highlight the advantages of utilizing the full signal values in energy landscape models for capturing neuronal dynamics, with strong implications for diagnosing and monitoring neurological disorders.
\end{abstract}


\section{Introduction}
\label{sec:introduction}

Understanding how the brain forms transient, task‑relevant circuits capable of coordinating information across distributed regions is central to explaining cognitive resilience and flexible control \cite{braver2003neural}. Biological systems must continually compress vast, high‑dimensional neural activity into lower‑dimensional latent representations, much like a conductor shaping the actions of diverse instruments, so that large‑scale networks can dynamically organize, stabilize, and reconfigure as needed \cite{dabagia2023aligning}. These latent structures are especially critical in networks such as the default mode (DMN), salience (SN), and limbic systems, where disruptions impair working memory, diminish executive control, and reduce cognitive resilience in patients with neurological disease \cite{kupis2021brain}. Capturing how these networks stabilize, interact, and transition among functional motifs is therefore of utmost importance and requires developing modeling frameworks that preserve the biological richness of neural signals while remaining computationally tractable.

In this regard, energy landscape models, inspired by statistical physics, have emerged as a fundamental framework for understanding brain dynamics \cite{watanabe2014energy,kang2019graph,olsen2024quality}. In these models, brain states are conceptualized as attractors within an energy landscape, with stable or metastable states represented as energy minima, and transitions between states reflecting the dynamics of neural activity \cite{ashourvan2017energy,ghaffari2025dynamic}. Despite their success, traditional Maximum Entropy Models (MEM) primarily rely on discretizing continuous neural signals, typically converting them into binary states ($\pm 1$) to simplify analysis and computations \cite{ezaki2017energy,kang2021bayesian}. However, binary discretization significantly reduces signal fidelity, discarding subtle yet biologically important variations in neural activity \cite{chen2021sources,taylor2023highlight}. In addition, such binary representation results in an exponential increase of the number of possible brain states as the number of regions of interest (ROIs) grows, making energy landscape calculations intractable even for moderately sized brain networks \cite{ezaki2017energy,kang2021bayesian}.

Hence, researchers have resorted to approximation methods to mitigate computational burdens in energy landscape calculation \cite{watanabe2013pairwise,ezaki2018ge,watanabe2014network}. For example, in \cite{ezaki2017energy}, The authors demonstrate that Maximum Pseudo-likelihood Estimation (PLE) can be used to approximate the likelihood function required for estimating the parameters of the maximum entropy model employed in energy landscape analysis. This method has become widely used due to its computational efficiency and asymptotic consistency. However, PLE often underestimates couplings in high-dimensional, noisy datasets due to its reliance on local approximations that ignore global dependencies. This limitation becomes more pronounced near critical points, where small changes in couplings lead to large variations in correlations \cite{kloucek2023biases}. 
Bayesian approaches, such as Bayesian maximum entropy model estimation (BMEM) \cite{kang2021bayesian}, have also been introduced that incorporate prior distributions to regularize parameter estimation and have shown to improve robustness, especially in scenarios with limited data \cite{kim2021variational}. Nonetheless, these methods still operate in binary states
potentially losing subtle but biologically meaningful signal fluctuations critical for capturing the inherent complexity of neural dynamics \cite{razavi2025brain,taylor2023highlight,chen2021sources}.

Here, we propose a novel continuous energy landscape
model using Graph Convolutional Networks (GCNs). Our
framework directly estimates the precision matrix from fMRI
signals without binarization, thereby preserving the full information inherent in neural signals. Additionally, we integrate bias terms analogous to external fields in traditional energy landscapes, capturing intrinsic, region-specific baseline activation differences that enhance model expressiveness and biological interpretability. We validated the proposed framework using a comprehensive set of experiments on both simulated data and real-world fMRI datasets from brain tumor patients. The results demonstrate significant improvements in parameter estimation and predictive modeling performance. These findings highlight the capability of our continuous energy landscape model to more effectively capture neural dynamics compared to traditional discretized approaches.

The rest of the paper is structured as follows: Section \ref{sec:preliminaries} provides the necessary background about conventional energy land scape models. Section \ref{sec:methods} introduces our proposed continuous GCN-based energy landscape framework including model formulation, GCN architecture, and optimization strategies. Section \ref{sec:experiments} describes experimental setup and datasets followed by Section \ref{sec:results} that presents the results and comparisons with traditional discrete energy landscape models. We conclude with a discussion of our proposed model including limitations and future directions in Section \ref{sec:discussion}. All the notations used in this paper are summarized in Table \ref{tab:notation} of Appendix \ref{sec:notations}.

\section{Preliminaries}
\label{sec:preliminaries}

\subsection{Discrete Energy Landscape Models}

The energy landscape framework conceptualizes whole-brain activity as a physical system in which each possible ``brain state" is assigned an energy value, reflecting its stability or likelihood of occurrence \cite{masuda2025energy}. The overall process of calculating an energy landscape involves three steps: defining the brain states, estimating the probability distribution of brain states, and quantifying the properties of the extracted distribution (i.e energy values). 

In the classical formulation of energy landscape analysis, we start with the continuous signal $\mathbf{y}$, thresholded to create the $N$ dimensional binary state vector $\mathbf{q}_{k} =(q_{k1},\dots,q_{kN})$, where $N$ is the total number of regions of interest (ROI) and $k$ is the index of $2^{N}$ distinct activation patterns. The indices $i$ and $j$ refer to individual ROIs, where $q_{ki}=+1$ if the signal is above the predefined activation threshold and $q_{ki}=-1$ otherwise. 

The pairwise maximum entropy model (MEM) then assigns a probability $p(\mathbf{q}_{k})$ to each state by maximizing the Shannon entropy subject to constraints that ensure the model reproduces the empirical first-order (mean activation) and second-order (pairwise correlations) moments observed in the data.

\begin{equation}
\begin{aligned}
\text{maximize:}&\quad
  -\sum_{k=1}^{2^{N}} p(\mathbf{q}_{k})\,
   \log p(\mathbf{q}_{k}) ,
\\
\text{subject to:}&\quad
  \sum_{k} p(\mathbf{q}_{k})\
       = 1,
\\ \quad 
  &\quad  \sum_{k} p(\mathbf{q}_{k})\,q_{ki}
       = \langle q_{i}\rangle_{\text{data}},
\\ \quad
  &\quad\sum_{k} p(\mathbf{q}_{k})\,q_{ki}q_{kj}
       = \langle q_{i}q_{j}\rangle_{\text{data}},
\end{aligned}
\end{equation}
where $\langle q_i \rangle_{\text{data}} = \frac{1}{T} \sum_{t=1}^T q_i^{(t)}$ and $\langle q_i q_j \rangle_{\text{data}} = \frac{1}{T} \sum_{t=1}^T q_i^{(t)} q_j^{(t)}$ denote the empirical average and pairwise correlation over the $T$ observed samples. The solution of this optimization problem has the Boltzmann (Gibbs) form:

\begin{equation}
p(\mathbf{q}_{k})
  \;=\;
  \frac{1}{\ell}\exp[-E(\mathbf{q}_{k})]
\label{eq:boltzmann_general}
\end{equation}

where $\ell={\displaystyle\sum_{k=1}^{2^{N}} \exp[-E(\mathbf{q}_{k})]}$ is the partition function that normalizes the probability distribution and $E$ is the energy function defined as: 
\begin{equation}
E(\mathbf{q}_{k})
  \;=\;
  -\sum_{i<j} W_{ij}\,q_{ki}q_{kj}
  \;-\;
  \sum_{i}   h_{i}\,q_{ki}.
\label{eq:bin_ener}
\end{equation}
Here, $W_{ij}$ captures the pairwise coupling between ROIs $i$ and $j$, while $h_{i}$ is a bias (external field) reflecting the intrinsic activity tendency of ROI~$i$. The pairwise sum $\sum_{i<j} W_{ij} q_i q_j$ is equivalent to the matrix form $-\tfrac12 \mathbf{q}^\top \mathbf{W} \mathbf{q}$, when $W_{ij}=W_{ji}$ and $W_{ii}=0$, with the factor $1/2$ avoiding double counting. Hence, the energy function in equation \ref{eq:bin_ener} could be rewritten in the compact form: 
\begin{equation}
E(\mathbf{q}) = -\frac{1}{2}\mathbf{q}^{\top}\mathbf{W}\mathbf{q} - \mathbf{h}^{\top}\mathbf{q}.
\label{eq:bin_ener_compact}
\end{equation}
\section{Methods}
\label{sec:methods}

\subsection{Continuous Energy Landscape Derivation}
\label{ssec:cont_energy}

In this study, we extend the conventional formulation of energy landscape analysis to directly utilize continuous-valued brain signals, denoted by $\mathbf{y} \in \mathbb{R}^{N}$, where $N$ represents the number of regions of interest (ROIs). This section outlines the key steps in deriving the continuous energy formulation, while the full mathematical details are provided in Appendix \ref{sec:appendix}.

To construct a continuous analog of the discrete energy landscape model, we replace the binary state vector $\mathbf{q} \in \{-1, +1\}^{N}$ with a continuous state vector $\mathbf{x} \in \mathbb{R}^{N}$, centered around a fixed baseline vector $\boldsymbol{\mu}$:

\begin{equation}
\mathbf{x} = \boldsymbol{\mu} + \tilde{\mathbf{y}},
\end{equation}

where $\tilde{\mathbf{y}}$ represents the signal fluctuations relative to the baseline $\boldsymbol{\mu}$. Using this continuous representation, we apply the maximum entropy principle to define a probability density function $p(\mathbf{x})$ by maximizing the differential Shannon entropy, subject to constraints on the empirical mean and covariance:

\begin{equation}
\begin{aligned}
\text{maximize:} &\quad
  -\int p(\mathbf{x}) \log p(\mathbf{x}) \, d\mathbf{x}, \\
\text{subject to:} &\quad
  \int p(\mathbf{x}) \, d\mathbf{x} = 1, \\
  &\quad \int p(\mathbf{x})\, x_i \, d\mathbf{x} = (\mu_{\Theta})_i, \\
  &\quad \int p(\mathbf{x})\, \big(x_i - (\mu_{\Theta})_i\big)\big(x_j - (\mu_{\Theta})_j\big) \, d\mathbf{x} = \Sigma_{ij},
\end{aligned}
\end{equation}

We parameterize via the precision $\mathbf{S} \succ 0$ and set $\boldsymbol{\Sigma}=\mathbf{S}^{-1}$.

Here, $(\mu_{\Theta})_i$ denotes the mean activity of region $i$, and $\Sigma_{ij}$ represents the empirical second-order moment (covariance) between regions $i$ and $j$. The solution to this optimization problem yields a multivariate Gaussian distribution with a Boltzmann-like form:

\begin{equation}
p(\mathbf{x}) = \frac{1}{\ell} \exp[-E(\mathbf{x})],
\label{eq:gaussian_boltzmann}
\end{equation}

where the partition function $\ell = (2\pi)^{N/2} |\boldsymbol{\Sigma}|^{1/2}\,\exp\!\big(-E_{\min}\big)$ normalizes the distribution, and $E_{\min}:=\min_{\mathbf{x}} E(\mathbf{x})$ denotes the global minimum of the energy, and the continuous energy function $E(\mathbf{x})$ is defined as:

\begin{equation}
E(\mathbf{x}) = \frac{1}{2} (\mathbf{x} - \boldsymbol{\mu})^\top \boldsymbol{\Sigma}^{-1} (\mathbf{x} - \boldsymbol{\mu}) - \mathbf{h}^\top \mathbf{x}.\
\label{eq:cel_energy}
\end{equation}

Here, $\boldsymbol{\Sigma}$ denotes the covariance matrix of the fluctuations $\tilde{\mathbf{y}}$, and $\mathbf{h}$ is a bias vector encoding the baseline activity levels of each ROI. Alternatively, we can express the distribution using the precision matrix $\mathbf{S} = \boldsymbol{\Sigma}^{-1}$:

\begin{equation}
p(\mathbf{x}) = \frac{\det(\mathbf{S})^{1/2}}{(2\pi)^{N/2}} \exp\left( -\tfrac{1}{2}\big(\mathbf{x}-\boldsymbol{\mu}_{\Theta}\big)^{\top}\mathbf{S}\big(\mathbf{x}-\boldsymbol{\mu}_{\Theta}\big) \right),
\label{eq:gaussian_boltzmann_S}
\end{equation}

where the continuous energy function is expressed as:  
\begin{equation}
E(\mathbf{x}) = \frac{1}{2} (\mathbf{x} - \boldsymbol{\mu})^\top \boldsymbol{S} (\mathbf{x} - \boldsymbol{\mu}) - \mathbf{h}^\top \mathbf{x}.
\label{eq:cel_energy_S}
\end{equation}


To ensure the probability density is properly normalized and integrable over $\mathbb{R}^N$, the quadratic form $(\mathbf{x} - \boldsymbol{\mu})^\top \mathbf{S} (\mathbf{x} - \boldsymbol{\mu})$ must be positive-definite. This guarantees that the energy function is bounded below and therefore the probability density in equation \ref{eq:gaussian_boltzmann} is well-defined and integrable. One way to enforce positive-definiteness is to approximate the precision matrix as $\mathbf{S} \equiv \alpha \mathbf{I} - \mathbf{W}$, where $\alpha > \lambda_{\max}(\mathbf{W})$ ensures that $\mathbf{S}$ is symmetric and positive-definite. However, this approach undermines the purpose of modeling a continuous energy function, since estimating the precision matrix via $\mathbf{S} = \alpha \mathbf{I} - \mathbf{W}$ relies on a coupling matrix $\mathbf{W}$ derived from discrete-state data. Therefore, in this study, we directly learn the precision matrix $\mathbf{S}$ from continuous data using Graph Convolutional Networks (see Sec.~\ref{ssec:gcn_param}).

\subsubsection{Unimodality of the quadratic CEL}
Because the precision matrix is positive definite ($\mathbf{S}\succ0$), the quadratic energy in Eq.~\eqref{eq:cel_energy_S} is strictly convex with a unique minimizer. Hence, this energy function defines a single basin, the set of states associated with that minimum, resulting in a unimodal distribution. However, many neurobiological systems exhibit multiple stable states corresponding to distinct functional configurations, which cannot be captured by a single basin. Therefore, here we use a mixture extension (CEL-Mix) in which $E(\mathbf{x})$ is the negative log of a Gaussian mixture (see section \ref{sec:GMM_Energy}).

\subsubsection{Boundedness of the energy landscape} Although the energy function includes a linear term $-\mathbf{h}^{\top}\mathbf{x}$, the quadratic term dominates for large $\|\mathbf{x}\|_2$, ensuring that $E(\mathbf{x})$ remains bounded from below.

By setting the gradient $\nabla E(\mathbf{x})=\mathbf{0}$, we find a unique minimizer at
\begin{equation}
\mathbf{x}_{\ast}=\boldsymbol{\mu}+\mathbf{S}^{-1}\mathbf{h},
\end{equation}
which provides the global energy minimum:
\begin{equation}
E_{\min}= -\mathbf{h}^{\top}\boldsymbol{\mu}-\frac{1}{2}\mathbf{h}^{\top}\mathbf{S}^{-1}\mathbf{h}.
\end{equation}
Consequently, the continuous energy landscape is strictly bounded from below:
\begin{equation}
E(\mathbf{x})\ge E_{\min}>-\infty,\quad \forall \mathbf{x}\in\mathbb{R}^{N},
\end{equation}
and the probability density \eqref{eq:gaussian_boltzmann} is well-defined and integrable.

\begin{lemma}[Boundedness and Gaussianity of the quadratic landscape]
Let $\mathbf{S}\in\mathbb R^{N\times N}$ be symmetric positive-definite matrix and let $\boldsymbol{\mu},\mathbf{h}$ be vectors in $\mathbb R^N$. Define $E(\mathbf{x})$ by \eqref{eq:cel_energy_S}.

Then $E$ is strictly convex with unique minimizer
$
\mathbf{x}_*=\boldsymbol{\mu}+\mathbf{S}^{-1}\mathbf{h},
$ and minimum value
$E_{\min}= -\mathbf{h}^\top \boldsymbol{\mu} - \tfrac12\,\mathbf{h}^\top \mathbf{S}^{-1}\mathbf{h}.$
Moreover,
\begin{equation}
E(\mathbf{x})=\tfrac12(\mathbf{x}-\mathbf{x}_*)^\top \mathbf{S} (\mathbf{x}-\mathbf{x}_*) + E_{\min},
\end{equation}
so $E$ is bounded below and coercive.  Thus, $p(\mathbf{x})\propto\exp[-E(\mathbf{x})]$ coincides precisely with the multivariate Gaussian \eqref{eq:gaussian_boltzmann_S}.
\end{lemma}

\begin{proof}
Completing the square yields the stated decomposition and the expression for $E_{\min}$. Because $\mathbf{S}\succ0$, the quadratic form $(\mathbf{x}-\mathbf{x}_*)^\top \mathbf{S} (\mathbf{x}-\mathbf{x}_*)\ge 0$ for all $\mathbf{x}$, with equality iff $\mathbf{x}=\mathbf{x}_*$. This implies strict convexity, boundedness from below with minimum $E_{\min}$, and the Gaussian form for $p(\mathbf{x})$.
\end{proof}

\subsection{Graph-Convolutional Parametrization of the Precision Matrix}
\label{ssec:gcn_param}

To estimate the precision matrix $\mathbf{S}$, we use a data-driven parameterization based on Graph Convolutional Networks (GCNs). In this setting, the continuous energy landscape model defines a probability density:
\begin{equation}
p(\mathbf{x} \mid \boldsymbol{\Theta}) \propto \exp\big[-E(\mathbf{x}; \boldsymbol{\Theta})\big],
\end{equation}

where $\boldsymbol{\Theta}$ denotes the trainable parameters of the model, including the GCN weights, the linear projection ($\mathbf{W}_Z$), and the mapping that produces the external field ($\mathbf{h}(\boldsymbol{\Theta})$). These parameters are learned by maximizing the likelihood of the observed data, which is equivalent to minimizing a loss function based on the negative log-likelihood (NLL). For $T$ samples, the likelihood is:

\begin{equation}
\mathcal{L}(\boldsymbol{\Theta}) = \prod_{t=1}^{T} p(\mathbf{x}_t \mid \boldsymbol{\Theta}),
\end{equation}

where we take its log to avoid numerical underflow and simplify computation:

\begin{equation}
\log \mathcal{L}(\boldsymbol{\Theta}) = \sum_{t=1}^{T} \log p(\mathbf{x}_t \mid \boldsymbol{\Theta}).
\end{equation}

Since optimization frameworks mainly minimize a loss function, we define:

\begin{equation}
J(\boldsymbol{\Theta}) = -\log \mathcal{L}(\boldsymbol{\Theta}),
\end{equation}

which is the negative log-likelihood that maximizes the likelihood function. Given that $p(\mathbf{x}|\boldsymbol{\Theta})$ is a multivariate Gaussian, the loss function has a closed form:
\begin{equation}
\begin{aligned}
J(\boldsymbol{\Theta}) &=
\frac{1}{2}\sum_{t=1}^{T}
\big(\mathbf{x}_{t}-\boldsymbol{\mu}_{\boldsymbol{\Theta}}\big)^{\top}
\mathbf{S}(\boldsymbol{\Theta})
\big(\mathbf{x}_{t}-\boldsymbol{\mu}_{\boldsymbol{\Theta}}\big) \\
&\quad
-\frac{T}{2}\log\det(\mathbf{S}(\boldsymbol{\Theta}))
+\frac{TN}{2}\log 2\pi,
\end{aligned}
\end{equation}

where $\boldsymbol{\mu}_{\boldsymbol{\Theta}}=\boldsymbol{\mu}+\mathbf{S}(\boldsymbol{\Theta})^{-1}\mathbf{h}(\boldsymbol{\Theta})$ is the effective mean of the Gaussian distribution. This shift in the mean arises from the completed-square form of Eq.~\eqref{eq:cel_energy} (see Appendix \ref{sec:appendix}), where $\mathbf{h}(\boldsymbol{\Theta})$ represents the external fields. The term $\tfrac{TN}{2}\log 2\pi$ is constant with respect to $\boldsymbol{\Theta}$ and therefore could be dropped to simplify the loss function. 

To ensure numerical stability and avoid ill-conditioning, we introduce a Frobenius regularization term and optimize the updated loss function:
\begin{equation}
J'(\boldsymbol{\Theta})=J(\boldsymbol{\Theta})+\lambda\|\mathbf{S}(\boldsymbol{\Theta})\|_{F}^{2},\quad\lambda>0,
\label{eq:final_loss}
\end{equation}

where $F$ is the Frobenius norm.  Here, the positive-definite precision matrix $\mathbf{S}$ can be approximated as:
\begin{equation}
\mathbf{S} = \mathbf{Z}\mathbf{Z}^{\top} + \epsilon\mathbf{I},
\label{eq:precisionmatrix}
\end{equation}

where $\epsilon = 10^{-5}$ is a small positive constant added to ensure invertibility, and $\mathbf{Z}$ is a low-rank embedding obtained using a shared linear map:
\begin{equation}
\mathbf{Z} = \mathbf{H}\,\mathbf{W}_Z \in \mathbb{R}^{N\times r},
\label{eq:lowrankembedding}
\end{equation}

where $\mathbf{H}\in\mathbb{R}^{N\times d}$ is the per-node embedding produced by the final GCN layer, $\mathbf{W}_Z(\boldsymbol{\Theta}) \in \mathbb{R}^{d\times r}$ is a trainable projection matrix, and $r \ll N$ is the rank of the low-dimensional representation. Both $\mathbf{H}$ and $\mathbf{W}_Z(\boldsymbol{\Theta})$ are learned jointly with the remaining parameters in $\boldsymbol{\Theta}$ by minimizing $J'(\boldsymbol{\Theta})$ via gradient-based optimization. We selected $r$ from $\{8,12,16\}$ using time-blocked 3-fold cross-validation (CV), which is specifically designed for time series data \cite{liu2024using}.

The closed-form loss function in equation~\eqref{eq:final_loss} involves $\det(\mathbf{S})$, whose direct computation is expensive ($O(N^3)$ for an $N \times N$ matrix). To reduce complexity, we apply the matrix determinant lemma, which states:
\begin{equation}
\det(\mathbf{S}) = \det(\epsilon \mathbf{I} + \mathbf{Z}\mathbf{Z}^{\top}) = \epsilon^{N} \big|\mathbf{I}_r + \epsilon^{-1}\mathbf{Z}^{\top}\mathbf{Z}\big|,
\end{equation}
where $\mathbf{I}_r$ is the $r \times r$ identity matrix. This reduces the computational cost from $O(N^3)$ to $O(r^3)$, which is significantly more efficient given that $r \ll N$.

Equation~\eqref{eq:final_loss} also includes the inverse of the precision matrix $\mathbf{S}^{-1}$, which can be computed efficiently using the Woodbury identity:
\begin{equation}
\mathbf{S}^{-1} = \epsilon^{-1} \mathbf{I} - \epsilon^{-2}\, \mathbf{Z} \!\left(\mathbf{I}_r+\epsilon^{-1} \mathbf{Z}^\top \mathbf{Z}\right)^{-1}\! \mathbf{Z}^\top.
\end{equation}

\subsection{Training Graph Convolutional Networks (GCNs)} \label{ssec:pipeline}

Figure~\ref{fig:pipeline_flow} presents the training pipeline for graph convolutional networks (GCNs) used to estimate the model parameters $\boldsymbol{\Theta}$. The process begins by extracting functional brain networks from fMRI time series data. We construct an undirected graph $G = (V, E)$, where each node corresponds to a region of interest (ROI). Let $\mathbf{R} \in \mathbb{R}^{N \times N}$ denote the sample Pearson correlation matrix computed from standardized time series, where each ROI time series is transformed to zero mean and unit variance by subtracting its temporal mean and dividing by its temporal standard deviation. We retain a fixed edge density $\delta \in (0, 1)$ by thresholding $|R_{ij}|$ at the $(1-\delta)$-quantile of the off-diagonal absolute correlations to ensure scale invariance and comparable sparsity across subjects:
\begin{equation}
    \tau = \operatorname{quantile}\!\big(\{|R_{ij}| : i < j\},\, 1-\delta\big),
\end{equation}
where $\tau$ is the correlation threshold. The adjacency matrix $\mathbf{A}$ is then defined as:
\begin{equation}
    A_{ij} = \mathbb{1}\!\left\{|R_{ij}| \ge \tau\right\},
\end{equation}
resulting in an undirected graph with edge density $\delta$. Unless otherwise specified, we set $\delta = 0.10$ (top 10\% absolute correlations) to ensure consistent sparsity across subjects for between-subject comparisons. We use the weighted adjacency matrix $B_{ij} = |R_{ij}|\,A_{ij}$ and its symmetric normalization:
\begin{equation}
    \tilde{\mathbf{B}} = \mathbf{D}^{-1/2} (\mathbf{I} + \mathbf{B}) \mathbf{D}^{-1/2},
\end{equation}
where $\mathbf{D}$ is the diagonal degree matrix of $(\mathbf{I} + \mathbf{B})$, i.e., $D_{ii} = \sum_j (\mathbf{I} + \mathbf{B})_{ij}$. This normalization adds self-loops and ensures scale invariance, preventing high-degree nodes from dominating the aggregation process in GCNs.

In the training phase, gradients of \eqref{eq:final_loss} are computed via automatic differentiation and optimized using the Adam optimizer (learning rate $10^{-4}$) with gradient clipping (threshold $0.1$). Training proceeds until convergence, defined as minimal changes in both loss and gradient norms (below $10^{-4}$) for five consecutive epochs.

Finally, once the optimal model parameters are obtained, we compute the embeddings $\mathbf{Z}(\Theta)$ from the final GCN layer and obtain 
$\mathbf{Z}(\boldsymbol{\Theta}) = \mathbf{H}(\boldsymbol{\Theta})\,\mathbf{W}_Z(\boldsymbol{\Theta}),$ 
where $\mathbf{W}_Z$ is a trainable projection matrix that is part of $\boldsymbol{\Theta}$ and learned jointly with the GCN weights and the external
field mapping by minimizing Eq.~\eqref{eq:final_loss}. The resulting $\mathbf{Z}(\boldsymbol{\Theta})$ is then used to compute $\mathbf{S}(\boldsymbol{\Theta})$
according to Eq.~\eqref{eq:precisionmatrix}.

\begin{figure*}[!t]
	\centering
	\includegraphics[width=0.95\linewidth]{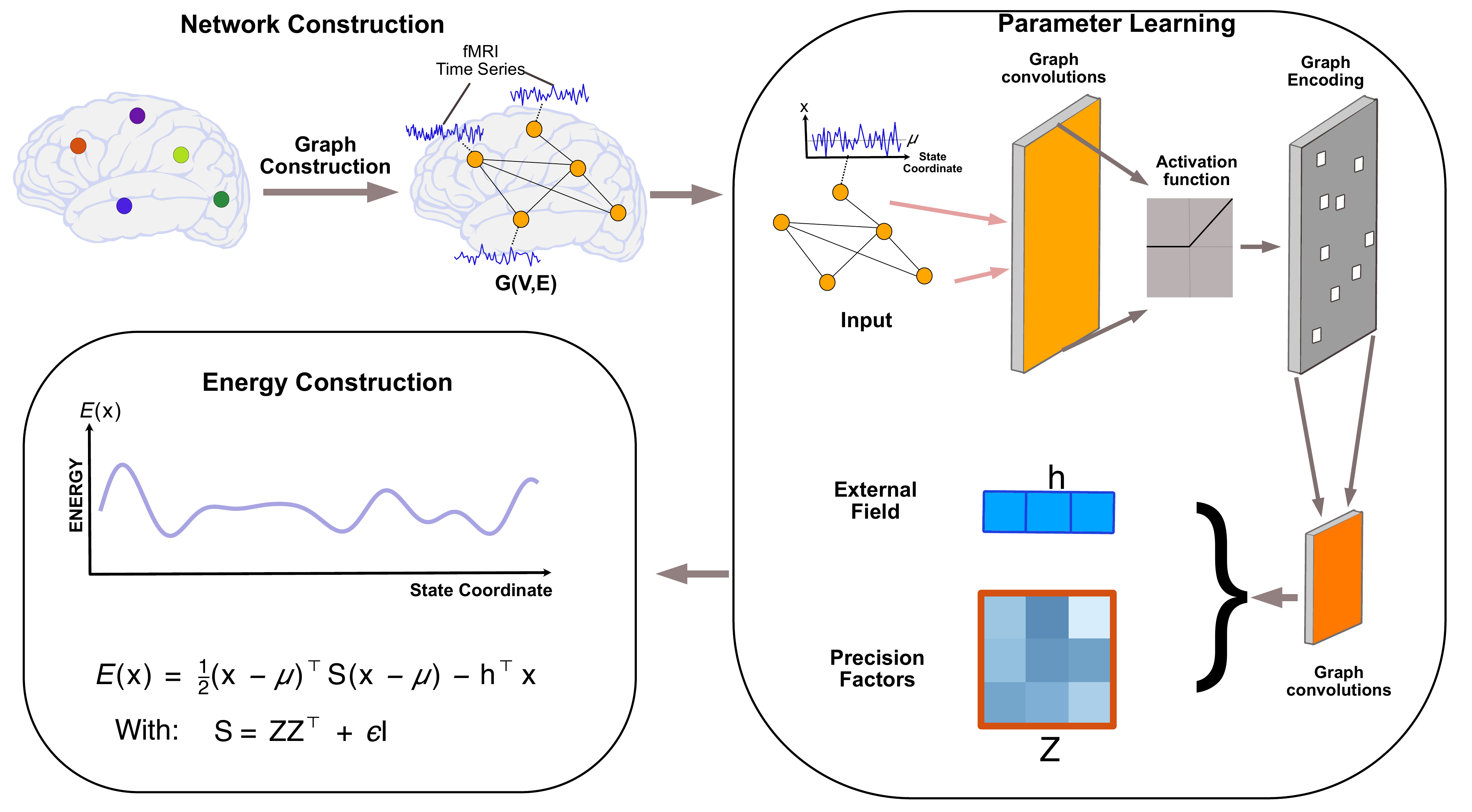}
	\caption{
        \textbf{Continuous energy landscape modeling via graph neural networks.} In Step one (Network Construction), we construct subject-specific functional networks $G=(V,E)$ by thresholding Pearson correlations from standardized fMRI time series and form the normalized weight matrix $\tilde{\mathbf{B}}$. In Step 2 (Parameter Learning), we estimate the precision matrix $\mathbf{S} = \mathbf{Z}\mathbf{Z}^\top + \epsilon \mathbf{I}$ together with the external field $\mathbf{h}(\boldsymbol{\Theta})$ by applying a graph convolutional network that produces node embeddings $\mathbf{H}(\boldsymbol{\Theta})$, which are projected via the trainable map $\mathbf{W}_Z(\boldsymbol{\Theta})$ to obtain embeddings $\mathbf{Z}$. Lastly, in the final step (Energy Construction), we utilize the learned $\mathbf{S}$, $\boldsymbol{\mu}_{\boldsymbol{\Theta}}$, and $\mathbf{h}(\boldsymbol{\Theta})$ to calculate the continuous energy values $E(\mathbf{x}_t;\boldsymbol{\Theta})$ and summarize the resulting energy landscape across time and subjects.
	}
	\label{fig:pipeline_flow}
\end{figure*}

\section{Experiments}
\label{sec:experiments}

We evaluated the performance of our proposed continuous energy landscape model in two settings. First, we utilized simulated data, which provided a known basin structure and transition dynamics, to validate the effectiveness of our proposed method. Next, we applied the model to real-world resting-state functional magnetic resonance imaging (rs-fMRI) data to investigate post-surgical cognitive outcomes in brain tumor patients. We also compared the performance of the continuous energy landscape model against the conventional discrete energy approach.

\subsection{Simulated Data}
\label{sec:simulated}

We simulated multivariate time-series data with a known multi-basin structure and a controlled transition matrix between basins. To ensure a balanced benchmark, we employed two distinct generators: (1) a Switching Linear Dynamical System (SLDS), which produces piecewise-linear Gaussian segments that align naturally with a continuous energy formulation, and (2) a Kuramoto oscillator network with template locking, which generates nonlinear, phase-coupled dynamics that yield clear attractor patterns, which aligns with a discrete, Ising-inspired formulation.

Using both ensures that neither modeling family is advantaged by a single data type and that each is tested in a typical operating regime. Throughout, $N$ denotes the number of variables (ROIs), $K$ the number of basins (states), $T$ the number of time points, and the signal-to-noise ratio (SNR) is defined as:
\begin{equation}
\mathrm{SNR} \;=\; \frac{\min_{k\neq k'}\|\boldsymbol{\mu}_k-\boldsymbol{\mu}_{k'}\|_2}{\sqrt{\mathrm{tr}(\boldsymbol{\Sigma})/N}},
\end{equation}
where $\boldsymbol{\mu}_k$ is the mean vector for state $k$ and $\boldsymbol{\Sigma}$ is the within-state covariance matrix used in SLDS. For the Kuramoto model, we first transform the phase variables to observables $\mathbf{y}_t = \sin(\boldsymbol{\varphi}_t)$, standardize these observables, and then compute $\boldsymbol{\mu}_k$ and $\boldsymbol{\Sigma}$ from $\mathbf{y}_t$. The SNR value controls the separation between basin centers relative to within-state variability, and we threshold it to define low, medium, and high separability regimes.

\subsubsection{Switching Linear Dynamical System (SLDS)}
To capture regime-switching behavior in time-series data, we combine a Markov chain for discrete state transitions with a Gaussian autoregressive model for continuous observations. Specifically, we define a first-order Markov chain with transition matrix $P_\star \in \mathbb{R}^{K \times K}$, where the diagonal entries $p_{\text{stay}} \in [0.80, 0.95]$ represent the probability of remaining in the same regime (dwell time). Given the latent path $(z_1,\ldots,z_T)$, the observations follow a Gaussian autoregressive process that interpolates the previous sample toward the mean of the active basin. More specifically, we initialize $\mathbf{x}_1 \sim \mathcal{N}(\boldsymbol{\mu}_{z_1},\boldsymbol{\Sigma})$ with the mean of the initial regime $\boldsymbol{\mu}_{z_1}$ and the positive definite covariance matrix $\boldsymbol{\Sigma}\succ0$. Next, for $t\ge2$ we add Gaussian noise $\boldsymbol{\xi}_t \sim \mathcal{N}(\mathbf{0},\boldsymbol{\Sigma})$ and update
\begin{equation}
\mathbf{x}_t \;=\; \rho\,\mathbf{x}_{t-1} \;+\; (1-\rho)\,\boldsymbol{\mu}_{z_t} \;+\; \boldsymbol{\xi}_t,
\end{equation}
where $\rho\in[0.2,0.5)$ is the autoregressive coefficient that induces autocorrelation while preserving memory of the previous observation. Basin centers $\{\boldsymbol{\mu}_k\}_{k=1}^{K}$ are spaced according to SNR so that basins are distinct yet partially overlapping. 

The outputs of this simulated model are continuous state vectors $\mathbf{X}=[\mathbf{x}_1,\ldots,\mathbf{x}_T]$, the ground-truth latent states $(z_t)$, means of dynamic regimes $\{\boldsymbol{\mu}_k\}$, and the transition matrix $P_\star$, which are utizlied to evaluated our proposed energy model.  

\subsubsection{Kuramoto Oscillator Network with Template Locking} \label{sec:kuramoto}
The Kuramoto model provides a framework for studying synchronization in networks of coupled oscillators. In neuroscience and complex systems, phase dynamics often exhibit structured patterns influenced by intrinsic properties and external constraints. To capture these dynamics, we extend the classical Kuramoto oscillator network by introducing template locking, which imposes weak alignment toward predefined phase configurations, analogous to externally forced Kuramoto models studied in the synchronization literature \cite{childs2008stability, wang2012exponential, yoon2021impact}.

We model the phase dynamics of the Kuramoto network using an Euler-Maruyama discretization of the underlying stochastic differential equation, with updates performed at time steps of size $\Delta t > 0$:
\begin{equation}
\begin{aligned}
\theta_{i,t+1} \;=\; \theta_{i,t} \;+\; \Delta t\Big(\omega_i \;+\; \sum_{j=1}^N C_{ij}\,\sin(\theta_{j,t}-\theta_{i,t})\Big) \\
\;+\; \Delta t\,\alpha\,\sin(\varphi_{z_t,i}-\theta_{i,t}) \;+\; \sqrt{2\zeta\,\Delta t}\,\xi_{i,t}.
\end{aligned} \label{eq:kuramoto}
\end{equation}

Here, $\theta_{i,t}$ is the phase of oscillator $i$ at time step $t$, $\omega_i$ is the natural frequency of oscillator $i$, $C_{ij}$ is the coupling strength between oscillators $i$ and $j$, $\alpha>0$ is the template phase-locking gain, $\varphi_{z_t,i}$ is the phase of oscillator $i$ in the selected template $\boldsymbol{\varphi}_{z_t}$ at time $t$, $\zeta$ is the diffusion coefficient controlling noise intensity, and $\xi_{i,t} \overset{\text{i.i.d.}}{\sim} \mathcal{N}(0,1)$ are independent standard Gaussian variables. The outputs of this Kuramoto model simulation are the continuous phase-derived observables $y_{i,t} = \sin(\theta_{i,t})$, which form the multivariate time series used as input to the energy landscape analysis.

\subsubsection{Extracting Energy Landscapes}
Here we discuss how we estimate the energy landscapes from data and assign each time point to a basin for both discrete and continuous formulations.

\paragraph{Discrete Energy Landscape}  
For the discrete energy landscape, we first binarize each variable's signal using its median, then estimate the Ising model parameters via $\ell_2$-regularized pseudolikelihood. Next, we identify local energy minima (basins) using a greedy single-variable update procedure, where we iteratively flip any binary variable (change its state from $-1$ to $+1$ or vice versa) that decreases the energy $E$ until no further improvement is possible. Finally, each time point is assigned to the basin reached by this energy descent.

\paragraph{Continuous Energy Landscape} \label{sec:GMM_Energy}

For the continuous energy model, we allow multiple basins in $\mathbb{R}^N$ by fitting a Gaussian mixture extension of our unimodal CEL described in Sec.~\ref{ssec:cont_energy}. In other words, we let each Gaussian component has its own positive definite precision matrix $\mathbf{S}_m$, and in the simulation experiments we set the number of mixture components to $M = K$ for comparability. In applications where the true number of basins is unknown, $M$ is treated as a hyperparameter and selected using standard model selection criteria, and we estimate the parameters $\{\eta_m,\boldsymbol{\mu}_m,\mathbf{S}_m \}_{m=1}^M$ by maximum likelihood, where $\eta_m$ denotes the mixture weight for component $m$ and satisfies $\eta_m \ge 0$ and $\sum_{m=1}^M \eta_m = 1$. The energy function is then defined as the negative log of the mixture density:
\begin{equation} \label{eq:GMM_Energy}
E(\mathbf{x}) = -\log\!\Bigg(\sum_{m=1}^{M}\eta_m\,\phi(\mathbf{x};\boldsymbol{\mu}_m,\mathbf{S}_m)\Bigg),
\end{equation}
where $\phi(\mathbf{x};\boldsymbol{\mu}_m,\mathbf{S}_m)$ denotes the multivariate Gaussian density with precision matrix $\mathbf{S}_m$:
\begin{equation}
\phi(\mathbf{x};\boldsymbol{\mu},\mathbf{S}) = \frac{|\mathbf{S}|^{1/2}}{(2\pi)^{N/2}} \exp\!\Big[-\tfrac{1}{2}(\mathbf{x}-\boldsymbol{\mu})^\top \mathbf{S} (\mathbf{x}-\boldsymbol{\mu})\Big].
\end{equation}

The energy function in \ref{eq:GMM_Energy} is generally nonconvex because it involves the negative logarithm of a sum of exponentials. This nonconvexity arises from the mixture formulation, where each Gaussian component contributes a separate mode. As a result, the energy landscape exhibits multiple local minima, typically located in regions where a single component density dominates the mixture. These regions correspond to points where the posterior probability of a given component, conditioned on the observation, is maximal. This observation underpins the labeling procedure, which assigns each data point to the component with the highest posterior probability (commonly referred to as its responsibility):  
\begin{equation}
\gamma_{t,m} = \frac{\eta_m\,\phi(\mathbf{x}_t;\boldsymbol{\mu}_m,\mathbf{S}_m)}{\sum_{m=1}^M \eta_m\,\phi(\mathbf{x}_t;\boldsymbol{\mu}_m,\mathbf{S}_m)},
\end{equation}
which represent the posterior probabilities of component membership. Each observation $\mathbf{x}_t$ is then assigned to the component with the highest responsibility according to the maximum a posteriori (MAP) rule:
\begin{equation}
\hat{z}_t = \arg\max_m \gamma_{t,m},
\end{equation}

where $\hat{z}_t$ is the maximum a posteriori (MAP) estimate of the latent component label for observation $\mathbf{x}_t$.

\subsubsection{Evaluation Metrics} \label{sec:cont_metrics}

To compare the extracted energy features with the true underlying structure of the system, we first need to align the recovered energy states to the ground truth. The MAP assignments partition the data into clusters corresponding to local minima in the energy landscape, referred to as recovered basins. These basins are represented in the data space $\mathbb{R}^N$. For each recovered basin $b$, we compute its centroid $\bar{\mathbf{y}}^{(b)}$ as the empirical mean of the data points $\{\mathbf{y}_t\}$ assigned to $b$. For SLDS, we match $\bar{\mathbf{x}}^{(b)}$ to the true centers $\{\boldsymbol{\mu}_k\}$, and for Kuramoto, we match it to $\{\mathbf{y}_k\}$, where $\mathbf{y}_k = \sin(\boldsymbol{\varphi}_k)$.

We then apply the Hungarian algorithm to a $K \times K$ cost matrix, where each entry is the Euclidean distance between a recovered basin centroid and a true center. For the Ising basins, we first map them into $\mathbb{R}^N$ by computing their data-space centroids before alignment. This procedure produces an optimal one-to-one matching between the recovered basins and the ground-truth states. Next, we utilize this alignment to compute three evaluation metrics as follows.

\paragraph{Basin Recovery (BR)}  

The first metric is calculated as the fraction of true basins that have a matched recovered centroid within a pre-specified Euclidean tolerance $\kappa$. This metric ranges from $0$ to $1$, where the value one corresponds to perfect recovery of all basins. Formally, we define:
\begin{equation}
BR = \frac{1}{K} \sum_{k=1}^{K} \mathbb{1}\Big( \big\| \bar{\mathbf{x}}^{(b_k)} - \boldsymbol{\mu}_k \big\|_2 \leq \kappa \Big),
\end{equation}
where $\mathbb{1}$ is the indicator function, $\bar{\mathbf{x}}^{(b_k)}$ is the recovered centroid matched to true basin $k$, and $\|\cdot\|_2$ denotes the Euclidean norm.

\paragraph{Transition Matrix Accuracy (TMA)}  

For the second metric, we estimated a transition matrix $\hat{P}$ from the recovered labels and compare it to the ground-truth matrix $P_\star$ using the average row-wise total variation agreement:
\begin{equation}
\mathrm{TMA} \;=\; 1 \;-\; \frac{1}{2K}\sum_{l=1}^{K}\left\|\hat{P}_{l:}-P_{\star\,l:}\right\|_1,
\end{equation}
where $\|\cdot\|_1$ denotes the element wise $\ell_1$ norm. This metric equals to one for exact recovery and approaches zero as the difference between the actual transition matrix and the estimated transition matrix increases.

\paragraph{State Distribution Agreement (SDA)}  
The final metric compares the empirical occupancy of recovered states with the true occupancy computed from the ground-truth state sequence $(z_t)$. The empirical occupancy $\hat{\nu}$ is defined as the proportion of time spent in each state:
\begin{equation}
\hat{\nu}_k = \frac{1}{T} \sum_{t=1}^T \mathbb{1}(z_t = k), \quad k = 1,\dots,K,
\end{equation}
where $\mathbb{1}(\cdot)$ is the indicator function, and $z_t$ is the latent state at time $t$. We then compute:
\begin{equation}
\mathrm{SDA} \;=\; 1-\tfrac{1}{2}\sum_{k=1}^{K} |\hat{\nu}_k-\nu_k|,
\end{equation}
where $|\cdot|$ denotes the absolute value, and $\nu_k$ is the true occupancy of state $k$. SDA ranges from $0$ to $1$, with values close to $1$ indicating accurate recovery of state frequencies. 

\subsubsection{Experimental grid, summarization, and statistical testing}
We vary $N \in \{6,7,\ldots,14\}$ (number of ROIs), $K \in \{3,4,5\}$ (number of basins), $T \in \{500,1000\}$ (length), and SNR (low, medium, high). Each condition is repeated $50$ times with new random seeds. Results are summarized as means with percentile bootstrap confidence intervals across repeats. For between-method comparisons (continuous versus discrete), metrics are paired by simulation unit (same $N,K,T$, SNR, and seed). We test whether the median paired difference equals zero using the two-sided Wilcoxon signed-rank test and present paired $p$ values. Where factor-specific families of tests are reported, we control the false discovery rate using the Benjamini–Hochberg procedure at $0.05$. This procedure ranks all p-values and determines which ones are significant while keeping the expected proportion of false discoveries below the given threshold.

\subsection{Brain Tumor Data}
\label{sec:realdata}

We utilized a publicly available brain tumor imaging dataset \cite{aerts2022pre} to apply the proposed continuous energy landscape model to predict post-surgery cognitive decline (e.g., working memory, executive function, and reaction time) using pre-surgical rs-fMRI data. The following subsections describe the dataset, the training model, and the evaluation metrics.

\subsubsection{Dataset Description}
We utilized a publicly available dataset comprising $20$ adult brain tumor patients, including $7$ with gliomas and $13$ with meningiomas, with rs-fMRI and neuropsychological assessments conducted before surgery and six months after resection \cite{aerts2022pre}. Imaging included T1-weighted anatomy and rs-fMRI with a repetition time of approximately 2100 to 2400 ms, an echo time of approximately 27 ms, and voxel size of $3 \times 3 \times 3$ mm$^3$. 

\subsubsection{Preprocessing}
The rs-fMRI preprocessing steps include motion correction, slice-timing correction, brain extraction, anatomical registration, normalization to Montreal Neurological Institute space, and high-pass temporal filtering with a 0.01 Hz cutoff. Tumor masks were applied prior to region extraction to avoid lesion signal. This step ensures that any signal coming from the lesion (tumor tissue) does not contaminate the measurements for healthy brain regions. Time series were then extracted from regions associated with the Default Mode Network (DMN), the Salience Network (SN), and the Limbic Network (LN), given their importance in cognitive function \cite{menon2011large}.

\subsubsection{Cognitive Measures}
We considered working memory, executive function, and reaction time in this study. Working memory was assessed by a Spatial Span Performance (SSP) score with a range of $2$ to $9$. Executive function was assessed by a composite score with a range of $2$ to $11$. We binarized the working memory and executive function scores into low and high categories based on established clinical thresholds \cite{tran2024high}, while reaction time was retained as a continuous variable.

\subsubsection{Energy Extraction}
For each participant and each functional network (DMN, SN, and LN), we independently fit both continuous and discrete models using that participant's rs-fMRI data. Specifically, for the discrete model, we applied median binarization per region, estimated Ising parameters via regularized pseudolikelihood, computed per-timepoint energy values, and assigned each time point to a basin. To assign each time point to an energy basin, we implemented a greedy energy-descent procedure: starting from the time point’s binary pattern, we iteratively flipped bits to reduce the energy until reaching a local minimum (mode). This local minimum served as the basin label, grouping time points that converged to the same attractor in the energy landscape. For the continuous model, we performed Gaussian maximum likelihood estimation without binarization, yielding per-timepoint energy values. Because the quadratic energy landscape is unimodal, we could not identify multiple basins directly from the continuous energy landscape. Instead, we applied a Gaussian-mixture extension of CEL (CEL-Mix), where the energy is defined as the negative log-likelihood of a mixture of \(M\) Gaussian components (see Eq.~\eqref{eq:GMM_Energy}). The number of components \(M\) was determined using the Bayesian Information Criterion (BIC).

Because each subject's energy model is fit independently, absolute energy scales are not directly comparable across subjects. We therefore rescale each subject's energy values separately before using them in subsequent analyses:
\begin{equation}
\begin{aligned}
\tilde E_t &= \frac{E_t - E_{\min}}{\sigma_E},\\
\end{aligned}
\end{equation}
where $E_{\min}=-\mathbf{h}^\top \boldsymbol{\mu} - \tfrac12  \mathbf{h}^\top \mathbf{S}^{-1}\mathbf{h}$ is the minimum energy value and $\sigma_E=\sqrt{\operatorname{tr}(\mathbf{S}^{-1})/N}$ is expected variability in energy values for that subject's model. Unless stated otherwise, the primary predictive features are per-timepoint normalized energy values $\{\tilde E_t\}$ averaged within each network.

\subsubsection{Predictive Modeling}
We used nested Leave-One-Subject-Out (LOSO) cross-validation to train a random forest model based on energy features extracted for each subject. In each outer fold, one subject was held out for testing, and a random forest was trained on the remaining subjects' features. Model hyperparameters, including the number of trees, maximum depth, and minimum samples per leaf, were tuned within an inner leave-one-subject-out loop. For classification, we averaged timepoint-level predicted probabilities to obtain a subject-level probability and computed Area Under the ROC Curve (AUC) from these values. For accuracy and $F_1$ score, we applied a threshold of $0.5$ to convert probabilities into binary labels. Lastly for the regression, we averaged timepoint-level predictions to obtain a single subject-level estimate.

\subsubsection{Performance Evaluation}
As mentioned in the previous section, for the classification problem (predicting working memory and executive function), we report accuracy, AUC, and $F_1$ scores. For the regression problem (predicting reaction time), we report RMSE and $R^2$. All values are computed at the subject level within the outer LOSO folds and summarized as mean~$\pm$~standard deviation across folds.

For the statistical comparison between the discrete and continuous energy models, we applied standard paired tests appropriate for each metric: DeLong's test for correlated ROC curves (AUC), McNemar's test for paired accuracy proportions, a paired permutation test for subject-level $F_1$ scores, and the Wilcoxon signed-rank test for per-subject absolute errors in regression (with $R^2$ reported descriptively). Within each cognitive outcome (WM, EF, RT) across the three networks, $p$-values were adjusted using the Benjamini-Hochberg procedure at $q = 0.05$ to control the false discovery rate (FDR).

\section{Results}
\label{sec:results}

\subsection{Continuous Energy Landscape Models More Accurately Capture Intrinsic Neural Dynamics}
\label{sec:results_simulated}

Figure~\ref{fig:slds_grouped} shows the results of applying continuous and Ising-based binary energy models to the simulated dataset from the Switching Linear Dynamical System (SLDS). Performance is evaluated across the three metrics (BR, TMA, SDA) described in section \ref{sec:cont_metrics}. As seen in the plots, the continuous energy landscape model provides significantly better performance compared to the discrete Ising-based model in capturing the underlying energy states and their transitions. The improvements are particularly noticeable in terms of transition matrix accuracy and state distribution agreement, which suggest that the continuous model more effectively captures the dynamics and structural properties of the latent energy landscape.

\begin{figure*}[!t]
    \centering
    \includegraphics[width=\textwidth]{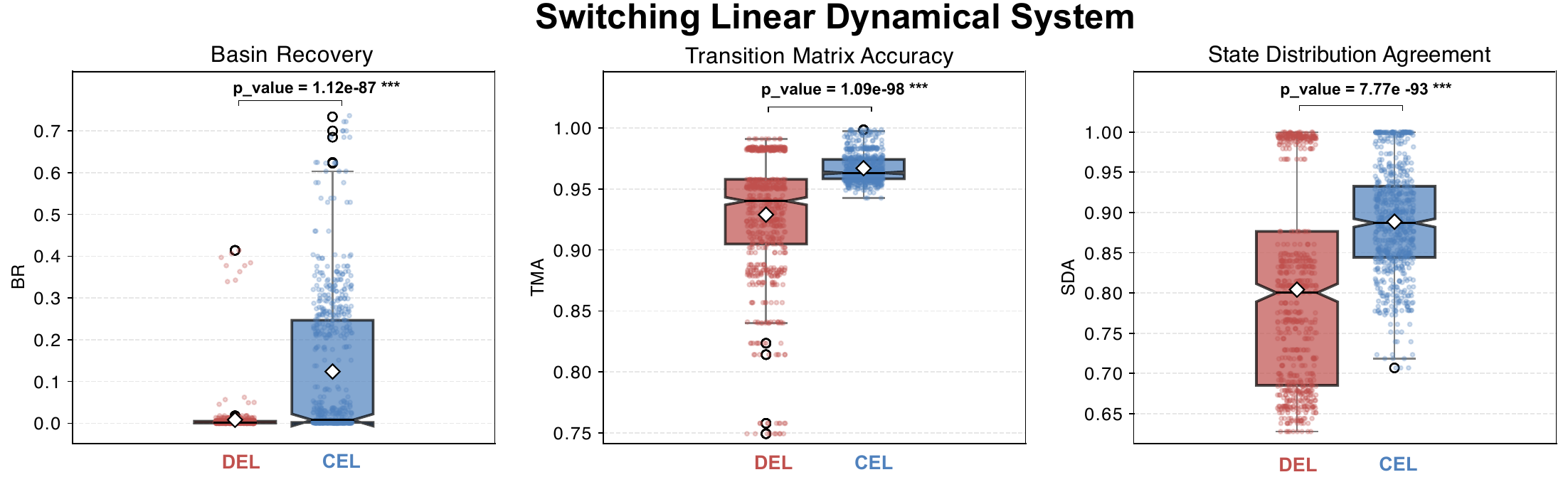}
    \caption{Switching Linear Dynamical System (SLDS): grouped performance on Basin Recovery (BR), Transition Matrix Accuracy (TMA), and State Distribution Agreement (SDA) for the discrete (DEL) and continuous (CEL) energy landscape models across the simulation grid ($N\in\{6,\ldots,14\}$, $K\in\{3,4,5\}$, $T\in\{500,1000\}$, SNR levels). The bar plots summarize results across repeats, where paired Wilcoxon signed-rank $p$-values shown at the top of each figure. Here, CEL denotes the mixture extension (CEL-Mix) used for multi-basin recovery. The CEL model provides a significantly better performance in capturing the underlying energy states and their transitions.}
    \label{fig:slds_grouped}
\end{figure*}

We also tested both models using the Kuramoto model, which provides a framework for studying the collective dynamics in networks of coupled oscillators. The results shown in Fig.~\ref{fig:kuramoto_grouped} indicate that the continuous model exhibits significantly better performance in terms of basin recovery, but not in the other two metrics. These findings were expected, as the Kuramoto model introduced in Section~\ref{sec:kuramoto} exhibits global phase alignment that tends to collapse into two dominant macrostates: a synchronized regime, where the coupling term $\sum_{j=1}^N C_{ij}\,\sin(\theta_{j,t}-\theta_{i,t})$ in Eq.~\ref{eq:kuramoto} strongly drives phases toward alignment, and a desynchronized regime, where this coupling influence is weak relative to intrinsic frequencies and noise. This emergent bistability effectively discretizes the energy landscape at the macroscopic level, making Ising-based models well-suited for capturing these transitions, while continuous models provide limited advantage in this context.

\begin{figure*}[!t]
    \centering
    \includegraphics[width=\textwidth]{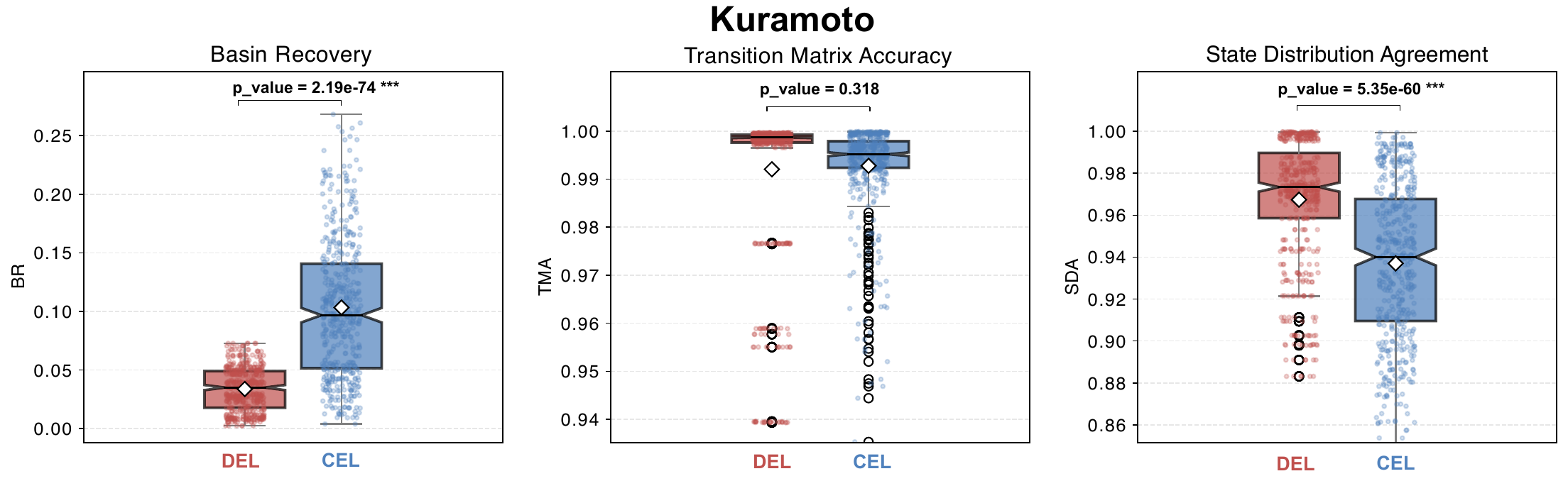}
    \caption{Kuramoto network with template locking: grouped performance on BR, TMA, and SDA for the discrete (DEL) and continuous (CEL) energy landscape models across the simulation grid. Bars summarize results across repeats; paired Wilcoxon signed-rank $p$-values are shown in-figure. Here, CEL denotes the mixture extension (CEL-Mix) used for multi-basin recovery. The CEL model provides a significantly better performance in terms of BR, but not for TMA and SDA. This was expected, as Ising-based discrete energy models are well-suited to capture the emergent bistability exhibited by Kuramoto dynamics.}
    \label{fig:kuramoto_grouped}
\end{figure*}

\subsection{Continuous Energy Landscape Features are More Predictive of Postsurgery Cognitive Outcomes}
\label{sec:results_realdata}

Table~\ref{tab:wm_results} shows the classification performance for predicting postoperative working memory outcomes using Ising-based and continuous energy features. The random forest model trained on continuous features achieves higher accuracy, area under the receiver operating characteristic curve (AUC), and $F_1$ scores across all canonical brain networks compared to the model trained on Ising-based features. The performance gain is dominant across networks (AUC $\geq 0.94$ for all three), with particularly large improvements in the default mode network (DMN) and limbic network (LN). Paired tests with Benjamini–Hochberg FDR correction confirmed these improvements: AUC-DMN/SN/Limbic all $p < 0.01$; Accuracy-significant across networks (DMN $p < 0.01$, SN $p < 0.05$, Limbic $p < 0.01$); and $F_1$-significant in all three (DMN $p < 0.05$, SN $p < 0.05$, Limbic $p < 0.05$).

The performance of the trained random forest models for predicting executive function is summarized in Table~\ref{tab:ec_results}. The model trained on continuous energy features consistently outperforms the model based on binarized energy features. Specifically, accuracy reaches $0.90$ in the Salience Network (SN), $0.85$ in the Default Mode Network (DMN), and similarly high values in the Limbic Network (LN), while AUC scores are as high as $0.97$. The $F_1$ scores remain balanced and high, indicating robust classification performance. After FDR correction, differences were statistically significant: AUC—DMN $p < 0.05$, SN $p < 0.05$, LN $p < 0.01$; Accuracy—DMN/SN/LN all $p < 0.01$; and $F_1$—DMN $p < 0.05$, SN/LN $p < 0.01$. 

Lastly, Table~\ref{tab:rt_results} shows regression performance for reaction time, where the random forest model trained on continuous energy features achieves lower RMSE and higher $R^2$ values across all networks, with $R^2$ reaching $0.46$ in the Limbic Network. In contrast, the model trained on Ising-based energy features yields substantially lower explanatory power, with $R^2$ values below $0.12$. Paired Wilcoxon tests on per-subject absolute errors confirmed the RMSE reductions (DMN/SN/Limbic all $p < 0.01$).

\begin{table*}[!t]
\centering
\begin{tabular}{llcccccc}
\toprule
\textbf{Network} & \textbf{Method} & \textbf{Accuracy} & \textbf{AUC} & \textbf{$F_1$} & $p_{\text{Acc}}$ & $p_{\text{AUC}}$ & $p_{F_1}$ \\
\midrule
\multirow{2}{*}{DMN}    & Discrete & 0.65 $\pm$ 0.108 & 0.67 $\pm$ 0.123 & 0.61 $\pm$ 0.112 & -- & -- & -- \\
                        & Continuous & \textbf{0.90 $\pm$ 0.095} & \textbf{0.94 $\pm$ 0.051} & \textbf{0.86 $\pm$ 0.165} & \textbf{$<$\!0.01} & \textbf{$<$\!0.01} & \textbf{$<$\!0.05} \\
\midrule
\multirow{2}{*}{SN}     & Discrete & 0.70 $\pm$ 0.120 & 0.79 $\pm$ 0.076 & 0.74 $\pm$ 0.095 & -- & -- & -- \\
                        & Continuous & \textbf{0.80 $\pm$ 0.113} & \textbf{0.96 $\pm$ 0.023} & \textbf{0.81 $\pm$ 0.109} & \textbf{$<$\!0.05} & \textbf{$<$\!0.01} & \textbf{$<$\!0.05} \\
\midrule
\multirow{2}{*}{Limbic} & Discrete & 0.75 $\pm$ 0.078 & 0.88 $\pm$ 0.052 & 0.76 $\pm$ 0.075 & -- & -- & -- \\
                        & Continuous & \textbf{0.90 $\pm$ 0.099} & \textbf{0.97 $\pm$ 0.032} & \textbf{0.88 $\pm$ 0.107} & \textbf{$<$\!0.01} & \textbf{$<$\!0.01} & \textbf{$<$\!0.05} \\
\bottomrule
\end{tabular}
\caption{\textbf{Working memory classification performance using energy features derived from discrete and continuous energy landscape models across the three networks (DMN, SN, and Limbic).} Reported metrics include the mean and standard deviation of Accuracy, AUC, and $F_1$ scores, computed using leave-one-subject-out cross-validation. Statistical comparisons include $p_{\text{AUC}}$ from DeLong's paired test for correlated ROC curves, $p_{\text{Acc}}$ from McNemar's test for paired accuracies, and $p_{F_1}$ from a paired permutation test, with significant values shown in bold. Overall, the continuous model consistently achieved higher AUC, accuracy, and $F_1$ scores compared to the discrete model.}
\label{tab:wm_results}
\end{table*}

\begin{table*}[!t]
\centering
\begin{tabular}{llcccccc}
\toprule
\textbf{Network} & \textbf{Method} & \textbf{Accuracy} & \textbf{AUC} & \textbf{$F_1$} & $p_{\text{Acc}}$ & $p_{\text{AUC}}$ & $p_{F_1}$ \\
\midrule
\multirow{2}{*}{DMN}    & Discrete & 0.70 $\pm$ 0.078 & 0.75 $\pm$ 0.084 & 0.76 $\pm$ 0.056 & -- & -- & -- \\
                        & Continuous & \textbf{0.85 $\pm$ 0.092} & \textbf{0.90 $\pm$ 0.102} & \textbf{0.85 $\pm$ 0.091} & \textbf{$<$\!0.01} & \textbf{$<$\!0.05} & \textbf{$<$\!0.05} \\
\midrule
\multirow{2}{*}{SN}     & Discrete & 0.70 $\pm$ 0.089 & 0.82 $\pm$ 0.079 & 0.70 $\pm$ 0.090 & -- & -- & -- \\
                        & Continuous & \textbf{0.90 $\pm$ 0.042} & \textbf{0.97 $\pm$ 0.101} & \textbf{0.90 $\pm$ 0.043} & \textbf{$<$\!0.01} & \textbf{$<$\!0.05} & \textbf{$<$\!0.01} \\
\midrule
\multirow{2}{*}{Limbic} & Discrete & 0.65 $\pm$ 0.046 & 0.68 $\pm$ 0.078 & 0.65 $\pm$ 0.052 & -- & -- & -- \\
                        & Continuous & \textbf{0.90 $\pm$ 0.082} & \textbf{0.97 $\pm$ 0.092} & \textbf{0.90 $\pm$ 0.084} & \textbf{$<$\!0.01} & \textbf{$<$\!0.01} & \textbf{$<$\!0.01} \\
\bottomrule
\end{tabular}
\caption{\textbf{Executive function classification performance using energy features derived from discrete and continuous energy landscape models across the three networks (DMN, SN, and Limbic).} Reported metrics include the mean and standard deviation of Accuracy, AUC, and $F_1$ scores, computed using leave-one-subject-out cross-validation. Statistical comparisons include $p_{\text{AUC}}$ from DeLong's paired test for correlated ROC curves, $p_{\text{Acc}}$ from McNemar's test for paired accuracies, and $p_{F_1}$ from a paired permutation test, with significant values shown in bold. The continuous model consistently achieved higher AUC, accuracy, and $F_1$ scores compared to the discrete model.}
\label{tab:ec_results}
\end{table*}

\begin{table*}[!t]
\centering
\begin{tabular}{llccc}
\toprule
\textbf{Network} & \textbf{Method} & \textbf{RMSE (ms)} & \textbf{$R^2$} & $p_{\text{RMSE}}$ \\
\midrule
\multirow{2}{*}{DMN}    & Discrete & 5651.173 $\pm$ 235 & 0.087 $\pm$ 0.052 & -- \\
                        & Continuous & \textbf{3904.849 $\pm$ 128} & \textbf{0.369 $\pm$ 0.034} & \textbf{$<$\!0.01} \\
\midrule
\multirow{2}{*}{SN}     & Discrete & 5556.262 $\pm$ 324 & 0.103 $\pm$ 0.075 & -- \\
                        & Continuous & \textbf{4377.112 $\pm$ 226} & \textbf{0.293 $\pm$ 0.042} & \textbf{$<$\!0.01} \\
\midrule
\multirow{2}{*}{Limbic} & Discrete & 5494.627 $\pm$ 227 & 0.113 $\pm$ 0.036 & -- \\
                        & Continuous & \textbf{3343.642 $\pm$ 112} & \textbf{0.460 $\pm$ 0.029} & \textbf{$<$\!0.01} \\
\bottomrule
\end{tabular}
\caption{\textbf{Reaction time regression performance using energy features derived from discrete and continuous energy landscape models across the three networks (DMN, SN, and Limbic).} Reported metrics include the mean and standard deviation of RMSE and $R^2$, computed using leave-one-subject-out cross-validation. Statistical comparisons include $p_{\text{RMSE}}$ values from a paired Wilcoxon signed-rank test on per-subject absolute errors, with significant values shown in bold. Overall, the continuous model consistently achieved lower RMSE and higher $R^2$ scores compared to the discrete model.}
\label{tab:rt_results}
\end{table*}

\section{Discussion}
\label{sec:discussion}
\subsection{Proposed Framework and Key Contributions}
Although neurons exhibit binary spiking behavior, large-scale brain dynamics are inherently continuous. Neural populations display graded firing rates and smooth fluctuations, and empirical evidence from neuroimaging and electrophysiology demonstrates that cognitive states evolve along continuous trajectories rather than abrupt transitions \cite{deco2017dynamics, breakspear2017dynamic, amini2025adhdeepnet}. In this study, we present a continuous formulation of energy landscape models to capture such gradual shifts and metastable dynamics in the brain. Our approach defines a continuous energy function based on the precision matrix and estimates this matrix using a data-driven framework built on graph convolutional networks with a custom loss function. To validate the proposed method, we conducted extensive simulation-based and real-world experiments, comparing its performance against traditional Ising-based energy models. 

Our results demonstrate that continuous energy landscape models consistently outperform discrete Ising-based models in capturing latent dynamics, recovering transition structures, and predicting cognitive outcomes in real-world fMRI datasets from brain tumor patients. These improvements are grounded in two well-established principles: (i) dichotomizing continuous variables discards valuable graded information and reduces statistical power \cite{maccallum2002practice}; and (ii) continuous models based on moment-matched Gaussian approximations offer greater tractability and efficiency for representing high-dimensional neuroimaging data \cite{hlinka2011functional, colclough2018multi}.

The continuous energy landscape model demonstrated strong performance in recovering state basins and transition dynamics within the Switching Linear Dynamical System (SLDS), accurately capturing both occupancy and switching patterns. In contrast, improvements offered by the continuous model were less pronounced in the Kuramoto system. This is because Kuramoto dynamics, while inherently continuous, exhibit strong pairwise phase correlations that align with the Ising model's assumption of binary interactions \cite{wang2021solving}, enabling the discrete model to capture certain aspects of synchronization despite its coarse representation. Nonetheless, overall the continuous model outperformed the Ising model in both experiments, highlighting its robustness and ability to preserve richer signal information.

\subsection{Insights on Cognitive Function}

In the energy landscape framework, each point in the state space represents a specific configuration of brain activity, and its associated energy value reflects the stability of that configuration \cite{watanabe2014energy}. Lower energy values indicate highly stable states, commonly interpreted as attractors corresponding to dominant cognitive or resting states, whereas higher energy values represent less stable configurations with an increased likelihood of transitions \cite{ashourvan2017energy}. The depth of an energy basin signifies the robustness of a state, where deeper basins suggest strong network integration and resistance to perturbations, while shallow basins indicate flexibility and a greater potential for rapid switching between states \cite{gu2018energy}. 

Within this context, our analysis of resting-state fMRI datasets indicates that preoperative neural stability, quantified through energy landscape models, is a strong predictor of postoperative cognitive function in brain tumor patients. This finding offers a mechanistic link between network stability and cognitive decline. Indeed, several studies have shown that functional connectivity patterns and the tumor’s proximity to specific large-scale networks are stronger predictors of cognitive outcomes compared to traditional anatomical features, such as tumor volume, which only provide moderate predictive power \cite{herbet2024predictors,luckett2024predicting}.

These results suggest that cognitive decline is not solely determined by the extent of structural damage but is critically influenced by the integrity and dynamic stability of functional networks. Hence, to better understand and predict postoperative outcomes, we need to utilize integrative approaches that combine structural and functional connectivity measures, where the stability and dynamic transitions of functional networks can be captured using energy landscape models. This integrated perspective could enable the development of robust biomarkers for surgical planning and personalized rehabilitation strategies aimed at preserving cognitive function.

\subsection{Limitations and Future Directions}
The proposed continuous energy landscape model represents an important step toward modeling the stability and dynamics of neuronal networks, but several challenges remain to improve its applicability and generalizability. First, the current formulation of the continuous energy landscape model relies on a Gaussian approximation, which uses the precision matrix to capture pairwise interactions among regions of interest. While this formulation enables efficient and tractable estimation of energy values, it is restricted to symmetric, undirected couplings and does not account for directional or higher-order interactions. Future extensions could incorporate asymmetric connectivity patterns and nonlinear dependencies to enable a more realistic representation of brain network transitions and their impact on cognitive function.

Furthermore, the current formulation of the continuous energy function employs graph neural networks (GNNs) to estimate the precision matrix. However, this indirect approach may introduce approximation errors and limit the model’s ability to capture complex dynamics. A more direct approach could integrate the continuous energy landscape formulation into the GNN loss function, enabling the network to directly estimate energy landscape values and bypass intermediate precision matrix estimation. Such an approach could allow the model to learn more efficient and detailed representations of network stability and transitions, improving predictive accuracy and interpretability.

Lastly, from a neurophysiological perspective, the current approach relies on functional connectivity without incorporating any structural information from the brain. Nonetheless, as discussed in the previous section, incorporating structural information could provide new insights into neural network dynamics and their relationship to cognitive function. Hence, in our future work, we will explore how structural connectivity can be integrated into the continuous energy landscape framework to create a more comprehensive and biologically plausible model. This integration could enhance the model’s ability to capture the interplay between structural connectivity patterns and functional dynamics, which in turn could improve interpretability and applicability of continuous energy landscape models in clinical settings.

\section{Conclusion}
\label{sec:conclusion}

In this work, we introduced a continuous energy landscape framework to model large-scale brain dynamics, addressing key limitations of traditional Ising-based models. By formulating the energy function using a Gaussian approximation based on the precision matrix, we leveraged graph convolutional networks to estimate this matrix, enabling us to capture graded neural fluctuations and metastable transitions more effectively. Through simulation analyses, we demonstrated that continuous energy models preserve richer signal information, accurately capture latent dynamics, and reliably recover transition structures, providing new insights into the underlying mechanisms of neural dynamics. Furthermore, using real-world fMRI datasets from brain tumor patients, we analyzed postoperative cognitive decline and showed that continuous energy values are more predictive of cognitive outcomes than conventional Ising-based energy values. Our results point to a promising new direction for continuous modeling of neural dynamics and its potential to improve predictive biomarkers for clinical decision-making. Such analyses could enable more accurate assessments of network stability and cognitive resilience, ultimately informing personalized surgical planning and rehabilitation strategies. Nonetheless, further studies are needed to validate these findings across larger cohorts and diverse neurological conditions for a more comprehensive understanding of the neural dynamics captured by energy functions and their relevance to cognitive outcomes.

\section{Code and Data availability}
All the code and data used in this study are available at: \\
\href{https://github.com/3sigmalab/ContinuousEnergyLandscape}{https://github.com/3sigmalab/ContinuousEnergyLandscape}

\IEEEpeerreviewmaketitle

\bibliographystyle{IEEEtran}
\bibliography{refs}

@preamble{ " \newcommand{\noop}[1]{} " }

@article{watanabe2014energy,
  title={Energy landscapes of resting-state brain networks},
  author={Watanabe, Takamitsu and Hirose, Satoshi and Wada, Hiroyuki and Imai, Yoshio and Machida, Toru and Shirouzu, Ichiro and Konishi, Seiki and Miyashita, Yasushi and Masuda, Naoki},
  journal={Frontiers in neuroinformatics},
  volume={8},
  pages={12},
  year={2014},
  publisher={Frontiers Media SA}
}

@article{kang2019graph,
  title={Graph-theoretical analysis for energy landscape reveals the organization of state transitions in the resting-state human cerebral cortex},
  author={Kang, Jiyoung and Pae, Chongwon and Park, Hae-Jeong},
  journal={PloS one},
  volume={14},
  number={9},
  pages={e0222161},
  year={2019},
  publisher={Public Library of Science San Francisco, CA USA}
}

@article{olsen2024quality,
  title={The quality and complexity of pairwise maximum entropy models for large cortical populations},
  author={Olsen, Valdemar Karg{\aa}rd and Whitlock, Jonathan R and Roudi, Yasser},
  journal={PLOS Computational Biology},
  volume={20},
  number={5},
  pages={e1012074},
  year={2024},
  publisher={Public Library of Science San Francisco, CA USA}
}

@article{ashourvan2017energy,
  title={The energy landscape underpinning module dynamics in the human brain connectome},
  author={Ashourvan, Arian and Gu, Shi and Mattar, Marcelo G and Vettel, Jean M and Bassett, Danielle S},
  journal={Neuroimage},
  volume={157},
  pages={364--380},
  year={2017},
  publisher={Elsevier}
}

@article{ezaki2017energy,
  title={Energy landscape analysis of neuroimaging data},
  author={Ezaki, Takahiro and Watanabe, Takamitsu and Ohzeki, Masayuki and Masuda, Naoki},
  journal={Philosophical Transactions of the Royal Society A: Mathematical, Physical and Engineering Sciences},
  volume={375},
  number={2096},
  pages={20160287},
  year={2017},
  publisher={The Royal Society Publishing}
}

@article{chen2021sources,
  title={Sources of information waste in neuroimaging: mishandling structures, thinking dichotomously, and over-reducing data},
  author={Chen, Gang and Taylor, Paul A and Stoddard, Joel and Cox, Robert W and Bandettini, Peter A and Pessoa, Luiz},
  journal={BioRxiv},
  pages={2021--05},
  year={2021},
  publisher={Cold Spring Harbor Laboratory}
}

@article{taylor2023highlight,
  title={Highlight results, don't hide them: Enhance interpretation, reduce biases and improve reproducibility},
  author={Taylor, Paul A and Reynolds, Richard C and Calhoun, Vince and Gonzalez-Castillo, Javier and Handwerker, Daniel A and Bandettini, Peter A and Mejia, Amanda F and Chen, Gang},
  journal={Neuroimage},
  volume={274},
  pages={120138},
  year={2023},
  publisher={Elsevier}
}

@article{kang2021bayesian,
  title={Bayesian estimation of maximum entropy model for individualized energy landscape analysis of brain state dynamics},
  author={Kang, Jiyoung and Jeong, Seok-Oh and Pae, Chongwon and Park, Hae-Jeong},
  journal={Human brain mapping},
  volume={42},
  number={11},
  pages={3411--3428},
  year={2021},
  publisher={Wiley Online Library}
}

@article{watanabe2013pairwise,
  title={A pairwise maximum entropy model accurately describes resting-state human brain networks},
  author={Watanabe, Takamitsu and Hirose, Satoshi and Wada, Hiroyuki and Imai, Yoshio and Machida, Toru and Shirouzu, Ichiro and Konishi, Seiki and Miyashita, Yasushi and Masuda, Naoki},
  journal={Nature communications},
  volume={4},
  number={1},
  pages={1370},
  year={2013},
  publisher={Nature Publishing Group UK London}
}

@article{ezaki2018ge,
  title={Age-related changes in the ease of dynamical transitions in human brain activity},
  author={Ezaki, Takahiro and Sakaki, Michiko and Watanabe, Takamitsu and Masuda, Naoki},
  journal={Human brain mapping},
  volume={39},
  number={6},
  pages={2673--2688},
  year={2018},
  publisher={Wiley Online Library}
}

@article{watanabe2014network,
  title={Network-dependent modulation of brain activity during sleep},
  author={Watanabe, Takamitsu and Kan, Shigeyuki and Koike, Takahiko and Misaki, Masaya and Konishi, Seiki and Miyauchi, Satoru and Miyahsita, Yasushi and Masuda, Naoki},
  journal={NeuroImage},
  volume={98},
  pages={1--10},
  year={2014},
  publisher={Elsevier}
}

@article{kloucek2023biases,
  title={Biases in inverse Ising estimates of near-critical behavior},
  author={Kloucek, Maximilian B and Machon, Thomas and Kajimura, Shogo and Royall, C Patrick and Masuda, Naoki and Turci, Francesco},
  journal={Physical Review E},
  volume={108},
  number={1},
  pages={014109},
  year={2023},
  publisher={APS}
}

@article{kim2021variational,
  title={Variational Bayes algorithm and posterior consistency of Ising model parameter estimation},
  author={Kim, Minwoo and Bhattacharya, Shrijita and Maiti, Tapabrata},
  journal={arXiv preprint arXiv:2109.01548},
  year={2021}
}

@article{aerts2022pre,
  title={Pre-and post-surgery brain tumor multimodal magnetic resonance imaging data optimized for large scale computational modelling},
  author={Aerts, Hannelore and Colenbier, Nigel and Almgren, Hannes and Dhollander, Thijs and Daparte, Javier Rasero and Clauw, Kenzo and Johri, Amogh and Meier, Jil and Palmer, Jessica and Schirner, Michael and others},
  journal={Scientific Data},
  volume={9},
  number={1},
  pages={676},
  year={2022},
  publisher={Nature Publishing Group UK London}
}

@inproceedings{tran2024high,
  title={High-Order Resting-State Functional Connectivity is Predictive of Working Memory Decline After Brain Tumor Resection},
  author={Tran, Triet M and Tran, Thi T and Khanmohammadi, Sina},
  booktitle={2024 46th Annual International Conference of the IEEE Engineering in Medicine and Biology Society (EMBC)},
  pages={1--5},
  year={2024},
  organization={IEEE}
}

@article{hlinka2011functional,
  title={Functional connectivity in resting-state fMRI: is linear correlation sufficient?},
  author={Hlinka, Jaroslav and Palu{\v{s}}, Milan and Vejmelka, Martin and Mantini, Dante and Corbetta, Maurizio},
  journal={NeuroImage},
  volume={54},
  number={3},
  pages={2218--2225},
  year={2011},
  publisher={Elsevier}
}

@article{razavi2025brain,
  title={Brain State Transition Disruptions in Alzheimer's Disease: Insights from EEG State Dynamics},
  author={Razavi, Seyed Majid and Tran, Triet M and Wilson, Steven E and Khanmohammadi, Sina},
  journal={2025 IEEE 22nd International Symposium on Biomedical Imaging (ISBI)},
  pages={1--5},
  year={2025},
  organization={IEEE}
}

@article{masuda2025energy,
  title={Energy landscape analysis based on the Ising model: Tutorial review},
  author={Masuda, Naoki and Islam, Saiful and Thu Aung, Si and Watanabe, Takamitsu},
  journal={PLOS Complex Systems},
  volume={2},
  number={5},
  pages={e0000039},
  year={2025},
  publisher={Public Library of Science San Francisco, CA USA}
}

@article{liu2024using,
  title={Using cross-validation methods to select time series models: Promises and pitfalls},
  author={Liu, Siwei and Zhou, Di Jody},
  journal={British Journal of Mathematical and Statistical Psychology},
  volume={77},
  number={2},
  pages={337--355},
  year={2024},
  publisher={Wiley Online Library}
}

@article{childs2008stability,
  title={Stability diagram for the forced Kuramoto model},
  author={Childs, Lauren M and Strogatz, Steven H},
  journal={Chaos: An Interdisciplinary Journal of Nonlinear Science},
  volume={18},
  number={4},
  year={2008},
  publisher={AIP Publishing}
}

@article{wang2012exponential,
  title={Exponential synchronization rate of Kuramoto oscillators in the presence of a pacemaker},
  author={Wang, Yongqiang and Doyle, Francis J},
  journal={IEEE transactions on automatic control},
  volume={58},
  number={4},
  pages={989--994},
  year={2012},
  publisher={IEEE}
}

@article{yoon2021impact,
  title={Impact of field heterogeneity on the dynamics of the forced Kuramoto model},
  author={Yoon, S and Wright, EAP and Mendes, JFF and Goltsev, AV},
  journal={Physical Review E},
  volume={104},
  number={2},
  pages={024313},
  year={2021},
  publisher={APS}
}

@article{menon2011large,
  title={Large-scale brain networks and psychopathology: a unifying triple network model},
  author={Menon, Vinod},
  journal={Trends in cognitive sciences},
  volume={15},
  number={10},
  pages={483--506},
  year={2011},
  publisher={Elsevier}
}

@article{maccallum2002practice,
  title={On the practice of dichotomization of quantitative variables.},
  author={MacCallum, Robert C and Zhang, Shaobo and Preacher, Kristopher J and Rucker, Derek D},
  journal={Psychological methods},
  volume={7},
  number={1},
  pages={19},
  year={2002},
  publisher={American Psychological Association}
}

@article{colclough2018multi,
  title={Multi-subject hierarchical inverse covariance modelling improves estimation of functional brain networks},
  author={Colclough, Giles L and Woolrich, Mark W and Harrison, Samuel J and L{\'o}pez, Pedro A Rojas and Valdes-Sosa, Pedro A and Smith, Stephen M},
  journal={NeuroImage},
  volume={178},
  pages={370--384},
  year={2018},
  publisher={Elsevier}
}

@article{deco2017dynamics,
  title={The dynamics of resting fluctuations in the brain: metastability and its dynamical cortical core},
  author={Deco, Gustavo and Kringelbach, Morten L and Jirsa, Viktor K and Ritter, Petra},
  journal={Scientific reports},
  volume={7},
  number={1},
  pages={3095},
  year={2017},
  publisher={Nature Publishing Group UK London}
}

@article{breakspear2017dynamic,
  title={Dynamic models of large-scale brain activity},
  author={Breakspear, Michael},
  journal={Nature neuroscience},
  volume={20},
  number={3},
  pages={340--352},
  year={2017},
  publisher={Nature Publishing Group}
}

@article{gu2018energy,
  title={The energy landscape of neurophysiological activity implicit in brain network structure},
  author={Gu, Shi and Cieslak, Matthew and Baird, Benjamin and Muldoon, Sarah F and Grafton, Scott T and Pasqualetti, Fabio and Bassett, Danielle S},
  journal={Scientific reports},
  volume={8},
  number={1},
  pages={2507},
  year={2018},
  publisher={Nature Publishing Group UK London}
}

@article{ghaffari2025dynamic,
  title={Dynamic fingerprinting of the human functional connectome},
  author={Ghaffari, Amin and Zhao, Yufei and Chen, Xu and Langley, Jason and Hu, Xiaoping},
  journal={bioRxiv},
  pages={2025--02},
  year={2025},
  publisher={Cold Spring Harbor Laboratory}
}

@article{amini2025adhdeepnet,
  title={ADHDeepNet From Raw EEG to Diagnosis: Improving ADHD Diagnosis through Temporal-Spatial Processing, Adaptive Attention Mechanisms, and Explainability in Raw EEG Signals},
  author={Amini, Ali and Alijanpour, Mohammad and Latifi, Behnam and Nasrabadi, Ali Motie},
  journal={arXiv preprint arXiv:2509.08779},
  year={2025}
}

@article{luckett2024predicting,
  title={Predicting post-surgical functional status in high-grade glioma with resting state fMRI and machine learning},
  author={Luckett, Patrick H and Olufawo, Michael O and Park, Ki Yun and Lamichhane, Bidhan and Dierker, Donna and Verastegui, Gabriel Trevino and Lee, John J and Yang, Peter and Kim, Albert and Butt, Omar H and others},
  journal={Journal of Neuro-Oncology},
  volume={169},
  number={1},
  pages={175--185},
  year={2024},
  publisher={Springer}
}

@article{herbet2024predictors,
  title={Predictors of cognition after glioma surgery: connectotomy, structure-function phenotype, plasticity},
  author={Herbet, Guillaume and Duffau, Hugues and Mandonnet, Emmanuel},
  journal={Brain},
  volume={147},
  number={8},
  pages={2621--2635},
  year={2024},
  publisher={Oxford University Press UK}
}

@article{wang2021solving,
  title={Solving combinatorial optimisation problems using oscillator based Ising machines},
  author={Wang, Tianshi and Wu, Leon and Nobel, Parth and Roychowdhury, Jaijeet},
  journal={Natural Computing},
  volume={20},
  number={2},
  pages={287--306},
  year={2021},
  publisher={Springer}
}

@article{braver2003neural,
  title={Neural mechanisms of transient and sustained cognitive control during task switching},
  author={Braver, Todd S and Reynolds, Jeremy R and Donaldson, David I},
  journal={Neuron},
  volume={39},
  number={4},
  pages={713--726},
  year={2003},
  publisher={Elsevier}
}

@article{dabagia2023aligning,
  title={Aligning latent representations of neural activity},
  author={Dabagia, Max and Kording, Konrad P and Dyer, Eva L},
  journal={Nature Biomedical Engineering},
  volume={7},
  number={4},
  pages={337--343},
  year={2023},
  publisher={Nature Publishing Group UK London}
}

@article{kupis2021brain,
  title={Brain dynamics underlying cognitive flexibility across the lifespan},
  author={Kupis, Lauren and Goodman, Zachary T and Kornfeld, Salome and Hoang, Stephanie and Romero, Celia and Dirks, Bryce and Dehoney, Joseph and Chang, Catie and Spreng, R Nathan and Nomi, Jason S and others},
  journal={Cerebral Cortex},
  volume={31},
  number={11},
  pages={5263--5274},
  year={2021},
  publisher={Oxford University Press}
}

\appendices
\section{Summary of Notations} \label{sec:notations}

\begin{table}[H]
\centering
\tiny
\caption{Summary of Notations.}
\label{tab:notation}
\resizebox{\linewidth}{!}{%
\begin{tabular}{ll}
\toprule
\textbf{Symbol} & \textbf{Meaning / Where used} \\
\midrule
$\|.\|_1$ & $L_1$ Norm \\
$\|.\|_2$ & $L_2$ (Euclidean) Norm\\
$\mathbb{1}$ & Indicator Function \\
$\succ0$ & Positive Definite Matrix \\
$\alpha$ & Template Phase Locking Gain\\
$\delta\in(0,1)$ & Edge Density \\
$\Delta t$ & Time Step \\
$\det(\cdot)$ & Determinant of a matrix \\
$\ell$ & Partition Function (Boltzmann Normalizer) \\
$\epsilon$ & Small Positive Constant used for Numerical Stability \\
$\eta_m$ &  Weight of Gaussian Mixture Component $m$ \\
$\gamma$  &  Responsibility of Gaussian Mixture Component (Posterior Probability)\\ 
$\kappa$ & Distance Threshold \\
$\lambda$ &  Regularization Parameter \\
$\mathcal{L}$ & Likelihood \\
$\mu$ & Signal Mean \\
$\nu$ & True Latent State Occupancy\\
$\hat{\nu}$ & Empirical State Occupancy \\
$\mu_{\Theta}$ & Effective Mean \\
$\boldsymbol{\mu}_{z}$ & Mean of Dynamic Regime\\
$\nabla$ & Gradient \\
$\omega$ & Natural Frequency of Oscillator \\
$\phi$ & Multivariate Gaussian Probability Density\\
$\rho$ & Autoregressive Coefficient \\
$\boldsymbol{\Sigma}$ & Signal Covariance\\
$\sigma_E$ & Energy Normalization Scale \\
$\tau$ & Correlation Threshold \\
$\theta$ & Phase in Kuramoto Model \\
$\boldsymbol{\Theta}$ & GCN Parameters \\
$\mathrm{tr}(\cdot)$ & Trace of a Matrix \\
$\boldsymbol{\xi}_t$ & Gaussian Noise \\
$\zeta$ &Diffusion Coefficient \\
$\mathbf{A}$ & Adjacency Matrix\\
$\mathbf{B}$ & Weighted Adjacency Matrix\\
$\tilde{\mathbf{B}}$& Normalized Weighted Adjacency Matrix \\
$\mathbf{C}$ & Kuramoto Coupling Matrix\\
$\mathbf{D}$ & Diagonal Matrix\\
$d$ & Node Feature Dimension \\
$d_M$ & Mahalanobis Distance \\
$E$ & Energy Value \\
$\tilde E$ & Normalized Energy Value \\
$E_{\min}$ & Minimum Energy Value \\
$f(\mathbf{x})$ & Auxiliary Energy Function \\
$G=(V,E)$ & Graph with Set of Nodes $V$ and edges $E$\\
$\mathbf{H}\in\mathbb{R}^{N\times d}$ & Hidden State in GCN \\
$\mathbf{h}$ & External Field (Bias) \\
$\mathbf{I}$ & Identity matrix \\
$i$ , $j$ & Index of ROIs \\
$J(\boldsymbol{\Theta})$ & Loss Function \\
$k$ & Index of Activation Pattern \\
$K$ & Total Number of Basins (States) \\
$l$ & The Row Index of Transition Matrix\\
$M$ & Number of Gaussian Mixture Components\\
$N$ & Total Number of Regions of Interest (ROIs)\\
$p(\mathbf{q}_{k})$ & Probability of State $\mathbf{q}$\\
$P_\star \in \mathbb{R}^{K \times K}$ & Transition Matrix\\
$\mathbf{q}_k\in\{-1,+1\}^N$ & Binary State Vector\\
$\mathbf{R}\in\mathbb{R}^{N\times N}$ & Pearson Correlation Matrix \\
$r$ & Latent Dimension\\
$\mathbf{S}\in\mathbb R^{N\times N}$ & Precision Matrix (Inverse Covariance)\\
$\mathrm{SNR}$ & Signal-to-Noise Ratio \\
$\mathrm{sym}(\mathbf{A})$ & Symmetric projection: $(\mathbf{A} + \mathbf{A}^\top)/2$ \\
$T$ & Total Number of Time Points\\
$t$ & Time Index\\
$\varphi$ & Template Phase in Kuramoto Model\\
$\mathbf{W}$ & Pairwise Coupling Matrix \\
$\mathbf{W}_Z$ & Latent Projection Matrix \\
$\mathbf{x}\in\mathbb R^N$ & Continuous State Vector\\
$\mathbf{x}_*$ & Global Minimizer\\
$\mathbf{y}\in\mathbb R^N$ & Continuous Signal\\
$\tilde{\mathbf{y}}\in\mathbb R^N$ & Signal Fluctuations \\
$\bar{\mathbf{y}}^{(b)}$ & Centroid of Basin $b$\\
$\mathbf{Z}$ & Low-Rank Embedding Matrix\\
$z$ & Latent Discrete State \\
$\hat{z}$ & Predicted Latent State \\  
\bottomrule
\end{tabular}}
\end{table}

\section{Continuous quadratic energy: derivations and properties}
\label{sec:appendix}



Consider the auxiliary quadratic form associated with the energy function $E(\mathbf{x})$, introduced to facilitate completing the square.
\[
f(\mathbf{x}) \;=\; \tfrac12\,\mathbf{x}^\top \mathbf{S}\,\mathbf{x} \;-\; \mathbf{h}^\top\mathbf{x},
\]

Since $\nabla f(\mathbf{x})=\mathbf{S}\mathbf{x}-\mathbf{h}$ and $\mathbf{S}\succ0$, the unique minimizer is
\begin{equation}
\mathbf{x}_*=\mathbf{S}^{-1}\mathbf{h}.
\end{equation}
Writing $\mathbf{x}=\mathbf{x}_*+\mathbf{u}$, we obtain
\begin{align}
f(\mathbf{x})
&=\tfrac12(\mathbf{x}_*+\mathbf{u})^\top\mathbf{S}(\mathbf{x}_*+\mathbf{u})-\mathbf{h}^\top(\mathbf{x}_*+\mathbf{u}) \\
&=\big(\tfrac12\mathbf{x}_*^\top\mathbf{S}\mathbf{x}_*-\mathbf{h}^\top\mathbf{x}_*\big)+\tfrac12\,\mathbf{u}^\top\mathbf{S}\mathbf{u}+\underbrace{\big(\mathbf{x}_*^\top\mathbf{S}\mathbf{u}-\mathbf{h}^\top\mathbf{u}\big)}_{=0}, \nonumber
\end{align}
because $\mathbf{S}\mathbf{x}_*=\mathbf{h}$. Hence
\begin{equation}
f(\mathbf{x})=f(\mathbf{x}_*)+\tfrac12(\mathbf{x}-\mathbf{x}_*)^\top\mathbf{S}(\mathbf{x}-\mathbf{x}_*).
\end{equation}

Compute
\begin{equation}
\begin{aligned}
f(\mathbf{x}_*)
&= \tfrac12\,\mathbf{x}_*^\top\mathbf{S}\,\mathbf{x}_* \;-\; \mathbf{h}^\top\mathbf{x}_* \\
&= \tfrac12\,\mathbf{h}^\top\mathbf{S}^{-1}\mathbf{h} \;-\; \mathbf{h}^\top\mathbf{S}^{-1}\mathbf{h} \\
&= -\tfrac12\,\mathbf{h}^\top\mathbf{S}^{-1}\mathbf{h}.
\end{aligned}
\label{eq:f_x_star}
\end{equation}

Therefore,
\begin{equation}
f(\mathbf{x})= -\tfrac12\,\mathbf{h}^\top\mathbf{S}^{-1}\mathbf{h} + \tfrac12(\mathbf{x}-\mathbf{S}^{-1}\mathbf{h})^\top\mathbf{S}(\mathbf{x}-\mathbf{S}^{-1}\mathbf{h}).
\end{equation}
Returning to the original energy on the uncentered variable,
\[
E(\mathbf{x})=f(\mathbf{x} - \boldsymbol{\mu})-\mathbf{h}^\top\boldsymbol{\mu},
\]
we obtain
\begin{equation}
\begin{aligned}
E(\mathbf{x})
&= \tfrac12\big(\mathbf{x}-\boldsymbol{\mu}-\mathbf{S}^{-1}\mathbf{h}\big)^\top
   \mathbf{S}\,\big(\mathbf{x}-\boldsymbol{\mu}-\mathbf{S}^{-1}\mathbf{h}\big) \\
&\quad - \tfrac12\,\mathbf{h}^\top\mathbf{S}^{-1}\mathbf{h} \;-\; \mathbf{h}^\top\boldsymbol{\mu}.
\end{aligned}
\label{eq:E_complete_square}
\end{equation}

we now demonstrate an alternative derivation by introducing and subtracting the term $\tfrac12\,\mathbf{h}^\top\mathbf{S}^{-1}\mathbf{h}$, which corresponds to the quadratic contribution of the minimizer $\mathbf{x}_*=\mathbf{S}^{-1}\mathbf{h}$. This manipulation isolates the centered quadratic term and makes the equivalence explicit. Specifically, add and subtract $\tfrac12\,\mathbf{h}^\top\mathbf{S}^{-1}\mathbf{h}$ to $f(\mathbf{x})$:
\[
f(\mathbf{x})=\Big[\tfrac12\,\mathbf{x}^\top\mathbf{S}\mathbf{x}-\mathbf{h}^\top\mathbf{x}+\tfrac12\,\mathbf{h}^\top\mathbf{S}^{-1}\mathbf{h}\Big]-\tfrac12\,\mathbf{h}^\top\mathbf{S}^{-1}\mathbf{h},
\]
and verify by expansion that the bracket equals
\[
\tfrac12(\mathbf{x}-\mathbf{S}^{-1}\mathbf{h})^\top\mathbf{S}(\mathbf{x}-\mathbf{S}^{-1}\mathbf{h}).
\]

The following expressions provide the closed-form characterization of the energy function, including its global minimizer, minimum value, and completed-square decomposition. These formulations reveal the convex structure of $E(\mathbf{x})$ and its dependence on $\mathbf{S}$, $\boldsymbol{\mu}$, and $\mathbf{h}$.

\begin{equation}
\mathbf{x}_* = \boldsymbol{\mu} + \mathbf{S}^{-1}\mathbf{h},
\end{equation}

\begin{equation}
E_{\min} = -\mathbf{h}^\top\boldsymbol{\mu} - \tfrac12\,\mathbf{h}^\top\mathbf{S}^{-1}\mathbf{h},
\end{equation}

\begin{equation}
E(\mathbf{x}) = \tfrac12(\mathbf{x} - \mathbf{x}_*)^\top\mathbf{S}(\mathbf{x} - \mathbf{x}_*) + E_{\min}.
\end{equation}

Because $\mathbf{S} \succ 0$, the energy $E(\mathbf{x})$ is strictly convex and coercive, meaning it is bounded below by $E_{\min}$ and grows unbounded as $\|\mathbf{x}\|_2 \to \infty$.


Using the Boltzmann form $p(\mathbf{x}) \propto \exp[-E(\mathbf{x})]$ and the completed-square decomposition in Eq.~\eqref{eq:E_complete_square}, we have
\begin{equation}
p(\mathbf{x}) = \exp\Big[-E_{\min} - \tfrac12(\mathbf{x}-\mathbf{x}_*)^\top \mathbf{S} (\mathbf{x}-\mathbf{x}_*)\Big].
\end{equation}

Gaussian integration yields the partition function (Boltzmann normalizer):
\begin{equation}
\ell = \int_{\mathbb{R}^N} \exp[-E(\mathbf{x})] \, d\mathbf{x}
= \exp[-E_{\min}] \, (2\pi)^{N/2} \, \det(\mathbf{S})^{-1/2}.
\end{equation}

Hence, the normalized density is
\begin{align}
p(\mathbf{x})
&= \frac{\det(\mathbf{S})^{1/2}}{(2\pi)^{N/2}}
   \exp\!\Big[-\tfrac12(\mathbf{x}-\mathbf{x}_*)^\top \mathbf{S} (\mathbf{x}-\mathbf{x}_*)\Big], \\
\mathbf{x}_* &= \boldsymbol{\mu} + \mathbf{S}^{-1}\mathbf{h}.
\end{align}

Equivalently,
\begin{equation}
p(\mathbf{x}) = \mathcal{N}\!\big(\boldsymbol{\mu} + \mathbf{S}^{-1}\mathbf{h}, \; \mathbf{S}^{-1}\big)
= \mathcal{N}\!\big(\boldsymbol{\mu}_\Theta, \; \mathbf{S}^{-1}\big).
\end{equation}

If $\mathbf{S}\succ0$, then $E(\mathbf{x})\to+\infty$ as $\|\mathbf{x}\|_2\to\infty$ (coercive), so $\ell<\infty$ and the Gaussian density on $\mathbb{R}^N$ is integrable. Conversely, an indefinite or singular quadratic does not yield an integrable Gaussian density on $\mathbb{R}^N$.

\end{document}